\documentclass[11pt]{article} 
\usepackage{fullpage}
\usepackage{amsmath,amsthm,amssymb,hyperref,cleveref,graphicx,url,rotating}
\usepackage{cancel}
\usepackage{stackrel}
 \usepackage{listings}
 \usepackage{algorithm2e}

\graphicspath{{./},{./Fig/}}

\lstdefinelanguage{Maxima}{
keywords={diff,ev,tex,solve,part,assume,sqrt,integrate,abs,inf,exp},
sensitive=true,
comment=[n][\itshape]{/*}{*/}
}

\graphicspath{{./}{Fig/}}

\newtheorem{Remark}{Remark}
\newtheorem{Definition}{Definition}
\newtheorem{Proposition}{Proposition}
\newtheorem{Corollary}{Corollary}
\newtheorem{Example}{Example}

\def\sp{\mathrm{sp}}
\def\diag{\mathrm{diag}}
\def\YF{\mathrm{YF}}
\def\SEF{\mathrm{SEF}}
\def\EF{\mathrm{EF}}
\def\matrixtwotwo#1#2#3#4{\left[\begin{array}{ll}#1 & #2 \cr #3 & #4\end{array}\right]}
\def\Aff{\mathrm{Aff}}
\def\GL{\mathrm{GL}}
\def\std{\mathrm{std}}
\def\calN{\mathcal{N}}
\def\JS{\mathrm{JS}}
\def\calP{\mathcal{P}}
\def\Shannon{\mathrm{Shannon}}	
\def\erf{\mathrm{erf}}	
\def\bd{\mathrm{bd}}
\def\matrixtwotwo#1#2#3#4{\left[\begin{array}{cc}#1 & #2\cr #3 & #4\end{array}\right]}
\def\OC{\mathrm{OC}}
\def\bartheta{{\bar\theta}}
\def\bareta{{\bar\eta}}
\def\supp{\mathrm{supp}}
\def\cl{\mathrm{cl}}
\def\Inner#1#2{\left\langle #1,#2\right\rangle}
\def\inner#1#2{\langle #1,#2\rangle}
 \def\calE{\mathcal{E}}

 \def\dP{\mathrm{d}P}
 \def\dQ{\mathrm{d}Q}
 \def\dmu{\mathrm{d}\mu}
\def\Burg{\mathrm{Burg}}
\def\KL{\mathrm{KL}}

\def\KL{\mathrm{KL}}
\def\alphastar{{\alpha^*}}
\def\Bi{\mathrm{Bi}}
\def\Mtr#1{\mathrm{tr}(#1)}

\def\calX{\mathcal{X}}
\def\bbR{\mathbb{R}}
\def\tr{\mathrm{tr}}
\def\st{\ :\ }
\def\Mdet#1{{\mathrm{det}\left(#1\right)}}
\def\eqdef{:=}
\def\calM{\mathcal{M}}
\def\calA{\mathcal{A}}
\def\ri{\mathrm{ri}}
\def\dom{\mathrm{dom}}
\def\ceil#1{\lceil{#1\rceil}}

\title{Revisiting Chernoff Information with Likelihood Ratio Exponential Families}

\author{Frank Nielsen\thanks{ORCID:0000-0001-5728-0726}\\
Sony Computer Science Laboratories Inc\\
Tokyo, Japan}  
 
\date{}

\begin{document}
\maketitle

\begin{abstract}
The Chernoff information between two probability measures is a statistical divergence measuring their deviation defined as their maximally skewed Bhattacharyya distance.
Although the Chernoff information was originally introduced for bounding the Bayes error in statistical hypothesis testing, the divergence found many other applications due to its empirical robustness property found in applications ranging from information fusion to quantum information.
From the viewpoint of information theory, the Chernoff information can also be interpreted as a minmax symmetrization of 
the Kullback--Leibler divergence.
In this paper, we first revisit the Chernoff information between two densities of a measurable Lebesgue space by considering the exponential families induced by their geometric mixtures: The so-called likelihood ratio exponential families.
Second, we show how to (i) solve exactly the Chernoff information between any two univariate Gaussian distributions or get a closed-form formula using symbolic computing,
(ii) report a closed-form formula of the Chernoff information of centered Gaussians with scaled covariance matrices
 and (iii) use a fast numerical scheme to approximate the Chernoff information between any two multivariate Gaussian distributions. 
\end{abstract}


\section{Introduction}   

\subsection{Chernoff information: Definition and related statistical divergences}

Let $(\calX,\calA)$ denote a measurable space~\cite{keener2010theoretical} with sample space $\calX$ and finite $\sigma$-algebra $\calA$ of events.
A measure $P$ is absolutely continuous with respect to  another measure $Q$ if $P(A)=0$ whenever $Q(A)=0$: 
$P$ is said dominated by $Q$ and written notationally for short as $P\ll Q$. 
We shall write $P\not\ll Q$ when $P$ is not dominated by $Q$.
When $P\ll Q$, we denote by $\frac{\dP}{\dQ}$ the Radon-Nikodym density~\cite{keener2010theoretical} of $P$ with respect to $Q$. 

Let us introduce some statistical distances like the Kullback-Leibler divergence, the Bhattacharyya distance or the Hellinger divergence  which have been proven useful is characterizing or bounding the probability of error in 
Bayesian statistical hypothesis testing~\cite{Kailath-1967,Torgersen-1991,nielsen2014generalized}.

The Kullback-Leibler divergence~\cite{cover1999elements} (KLD) between two probability measures (PMs) $P$ and $Q$ is defined as
$$
D_\KL[P:Q]=\left\{\begin{array}{ll}
\int_\calX \log\left(\frac{\dP}{\dQ}\right)\, \dP, & \mbox{if $P\ll Q$}\\
+\infty & \mbox{if $P\not\ll Q$}
\end{array}\right.
$$
Two PMs $P$ and $Q$ are mutually singular when there exists an event $A\in\calA$ such that $P(A)=0$ and $Q(\calX\backslash A)=0$.
Mutually singular measures $P$ and $Q$ are notationally written as $P\perp Q$.
Let $P$ and $Q$ be two non-singular probability measures on $(\calX,\calA)$ dominated by a common $\sigma$-finite measure $\mu$, and denote by
 $p=\frac{\dP}{\dmu}$ and $q=\frac{\dQ}{\dmu}$ their Radon-Nikodym densities with respect to $\mu$.
Then the KLD between $P$ and $Q$ can be calculated equivalently by the KLD between their densities as follows:
\begin{equation}\label{eq:kld}
D_\KL[P:Q]=D_\KL[p:q]=\int_\calX p\log\left(\frac{p}{q}\right)\, \dmu.
\end{equation}
It can be shown that $D_\KL[p:q]$ is independent of the chosen dominating measure $\mu$~\cite{Torgersen-1991},
and thus when $P,Q\ll \mu$, we write for short $D_\KL[P:Q]=D_\KL[p:q]$.
Although the dominating measure $\mu$ can be set to $\mu=\frac{P+Q}{2}$ in general, it is either often chosen as $\mu_L$ the Lebesgue measure for continuous sample spaces $\bbR^d$ (with  the $\sigma$-algebra $\calA=\mathcal{B}(\bbR^d)$ of Borel sets) or as the counting measure $\mu_{\#}$ for discrete sample spaces (with the $\sigma$-algebra $\calA$ of power sets).
The KLD is not a metric distance because it is asymmetric and does not satisfy the triangle inequality.
 
Let $\supp(\mu)=\cl\{A\in\calA \st \mu(A)\not=0\}$ denote the support of a Radon positive measure~\cite{keener2010theoretical} $\mu$ where $\cl$ denotes the topological closure operation.
Notice that $D_\KL[p:q]=+\infty$ when the definite integral of Eq.~\ref{eq:kld} divergences (e.g., the KLD between a standard Cauchy distribution and a standard normal distribution is $+\infty$ but the KLD between a standard normal distribution and a standard Cauchy distributions is finite), and $D_\KL[P:Q]=\infty$ when the probability measures have disjoint supports ($P\perp Q$).
Thus when the supports of $P$ and $Q$ are distinct but not nested, both the forward KLD $D_\KL[P:Q]$ and the reverse KLD $D_\KL[Q:P]$ are infinite.

The Chernoff information~\cite{chernoff1952measure}, also called Chernoff information number~\cite{csiszar1972class,Torgersen-1991} or the Chernoff divergence~\cite{audenaert2007discriminating,audenaert2008asymptotic}, is the following {\em symmetric} measure of dissimilarity between any two comparable probability measures $P$ and $Q$ dominated by $\mu$:
$$
D_C[P,Q]\eqdef \max_{\alpha\in (0,1)} -\log  \rho_\alpha[P:Q] = D_C[Q,P],
$$
where
\begin{equation}\label{eq:GBC}
\rho_\alpha[P:Q]\eqdef \int p^\alpha q^{1-\alpha}\dmu=\rho_{1-\alpha}[Q:P],
\end{equation}
is the $\alpha$-skewed Bhattacharyya affinity coefficient~\cite{bhattacharyya1943measure} (a coefficient measuring the similarity of two densities).
The $\alpha$-skewed Bhattacharyya coefficients are always upper bounded by $1$  and  are  strictly greater than zero for non-empty intersecting support (non-singular PMs):
$$
0<\rho_\alpha[P:Q]\leq 1.
$$
A proof can be obtained by applying H\"older's inequality (see also Remark~\ref{rmk:BhatCoeff} for an alternative proof).

Since the affinity coefficient $\rho_\alpha[P:Q]$ does not depend on the underlying dominating measure $\mu$~\cite{Torgersen-1991}, we shall write $D_C[p,q]$ instead of $D_C[P,Q]$ in the reminder. 

Let $D_{B,\alpha}[p:q]$ denote the $\alpha$-skewed Bhattacharyya distance~\cite{bhattacharyya1943measure,nielsen2011burbea}:
$$
D_{B,\alpha}[p:q]\eqdef -\log \rho_\alpha[P:Q]= D_{B,1-\alpha}[q:p],
$$
The $\alpha$-skewed Bhattacharyya distances are not  metric distances since they can be asymmetric and do not satisfy the triangle inequality even when $\alpha=\frac{1}{2}$.

Thus the Chernoff information is defined as the maximal skewed Bhattacharyya distance:
\begin{equation}\label{eq:defCI}
D_C[p,q]=\max_{\alpha\in (0,1)} D_{B,\alpha}[p:q].
\end{equation}

Gr\"unwald~\cite{grunwald2007minimum,ITEF-2007} called the skewed Bhattacharyya coefficients and distances the $\alpha$-R\'enyi affinity and the unnormalized R\'enyi divergence, respectively (see \S 19.6 of~\cite{grunwald2007minimum}) since the R\'enyi divergence~\cite{RenyiDiv-2014} is defined by
$$
D_{R,\alpha}[P:Q]=\frac{1}{\alpha-1}\log \int p^\alpha q^{1-\alpha}\dmu = \frac{1}{1-\alpha} D_{B,\alpha}[P:Q].
$$
Thus $D_{B,\alpha}[P:Q]=(1-\alpha)\, D_{R,\alpha}[P:Q]$ can be interpreted as the unnormalized R\'enyi divergence in~\cite{grunwald2007minimum}.
However, let us notice that the R\'enyi $\alpha$-divergences are defined in general for a wider range 
 $\alpha\in [0,\infty]\backslash\{1\}$ with 
$\lim_{\alpha\rightarrow 1} D_{R,\alpha}[P:Q]=D_\KL[P:Q]$ but the skew Bhattacharyya distances are defined for $\alpha\in(0,1)$ in general.

The Chernoff information was originally introduced to upper bound the probability error of misclassification in Bayesian binary hypothesis testing~\cite{chernoff1952measure} where 
the optimal skewing parameter $\alphastar$ such that $D_C[p,q]=D_{B,\alphastar}[p:q]$ is referred to in the statistical literature as the Chernoff error exponent~\cite{cover1999elements,borade2006projection,boyer2017error} or Chernoff exponent~\cite{d2004distributed,yu2022comments} for short.
The Chernoff information has found many other fruitful applications beyond its original statistical hypothesis testing scope like in computer vision~\cite{konishi1999fundamental}, information fusion~\cite{julier2006empirical}, time-series clustering~\cite{kakizawa1998discrimination}, and more generally in machine learning~\cite{dutta2020there} (just to cite a few use cases).
It has been observed empirically that the Chernoff information exhibits superior robustness~\cite{agarwal2019limits} 
compared to the Kullback--Leibler divergence in
distributed fusion of Gaussian Mixtures Models~\cite{julier2006empirical} (GMMs) or in target detection in radar sensor network~\cite{maherin2014radar}. 
The Chernoff information has also been used for analysis deepfake detection performance of Generative Adversarial Networks~\cite{agarwal2019limits} (GANs).

\begin{Remark}\label{rmk:BhatCoeff}
Let $f_\alpha(u)=u^\alpha$ for $\alpha\in\bbR$. 
The functions $f_\alpha(u)$ are convex for $\alpha\in \bbR\backslash [0,1]$ and concave for $\alpha\in [0,1]$. 
Thus we can define the $f$-divergences~\cite{csiszar1963information,ali1966general} $I_{f_\alpha}[p:q]=\int p f_\alpha(q/p)\dmu=\int p^\alpha q^{1-\alpha}\dmu$ for $\alpha\in \bbR\backslash [0,1]$ and $I_{-f_\alpha}[p:q]=-\int pf_\alpha(q/p)\dmu=-\int p^\alpha q^{1-\alpha}\dmu$ for $\alpha\in (0,1)$  (or equivalently take the convex generator $h_\alpha(u)=-u^\alpha$ for $\alpha\in(0,1)$). 
Notice that the conjugate $f$-divergence is obtained for the generator 
$f_\alpha^*(u)=uf_\alpha(1/u) =u^{1-\alpha}$: $I_{f_\alpha}[q:p]=I_{f_\alpha^*}[p:q]$. 
By Jensen's inequality, we have that the $f$-divergences are lower bounded by $f(1)$.
Thus $I_{h_\alpha^*}[p:q]\geq h_\alpha(1)=-1$.
Since $f$-divergences are upper bounded by $f(0)+f^*(0)$, we have  that $I_{h_\alpha^*}[p:q]<0$ for $\alpha\in (0,1)$.
This gives another proof that the Bhattacharyya coefficient $\rho_\alpha[p:q]=-I_{h_\alpha}[p:q]$ is bounded between $0$ and $1$ since the $I_{h_\alpha}$ divergence is bounded between $-1$ and $0$.
Moreover, Ali and Silvey~\cite{ali1966general} further defined the $(f,g)$-divergences as $I_{f,g}[p:q]=g(I_f[p:q])$ for a strictly monotonically increasing function $g(v)$.
Letting $g(v)=-\log(-v)$ (with $g'(v)=-\frac{1}{v}<0$ when $v\in(0,1)$), we get that the $(h_\alpha,g)$-divergences are the Bhattacharyya distances for $\alpha\in (0,1)$.
However, the Chernoff information is not a $f$-divergence despite the fact that  Bhattacharyya distances are Ali-Silvey $(f,g)$-divergences because of the maximization criterion~\cite{ali1966general} of Eq.~\ref{eq:defCI}. 
\end{Remark}

\subsection{Prior work and contributions}
The Chernoff information between any two categorical distributions (multinomial distributions with one trial also called ``multinoulli'' since they are extensions of the Bernoulli distributions) has been very well-studied and described in many reference textbooks of information theory or statistics (e.g., see~Sec.~12.9 of~\cite{cover1999elements}). 
The Chernoff information between two probability distributions of an exponential family was considered from the viewpoint of information geometry in~\cite{CI-2013}, and in the general case from the viewpoint of unnormalized R\'enyi divergences in~\cite{RenyiDiv-2014} (Theorem~32).
By replacing the weighted geometric mean in the definition of the Bhattacharyya coefficient $\rho_\alpha$ of Eq.~\ref{eq:GBC} by an arbitrary weighted mean, the generalized Bhattacharyya coefficient and its associated divergences including the Chernoff information was generalized in~\cite{nielsen2014generalized}.
The geometry of the Chernoff error exponent was studied in~\cite{westover2008asymptotic,nielsen2013hypothesis} when dealing with a finite set of  mutually absolutely  probability distributions $P_1,\ldots, P_n$.
In this case, the Chernoff information amounts to 
the minimum pairwise Chernoff information of the probability distributions~\cite{leang1997asymptotics}:
$$
D_C[P_1,\ldots,P_n]:=\min_{i\in\{1,\ldots, n\}\not=j\in\{1,\ldots, n\}} D_C[P_i,P_j].
$$

We summarize our contributions as follows:
In section~\ref{sec:CILREF}, we study the Chernoff information between two given mutually non-singular probability measures $P$ and $Q$  by considering their ``exponential arc''~\cite{cena2007exponential} as a special 1D exponential family termed a Likelihood Ratio Exponential Family (LREF) in~\cite{ITEF-2007}.
We show that the optimal skewing value (Chernoff exponent) defining their Chernoff information is unique (Proposition~\ref{prop:CIalphaunique}) and can be characterized geometrically on the Banach vector space $L^1(\mu)$ of equivalence classes of  measurable functions  (i.e., two functions $f_1$ and $f_2$ are said equivalent in $L^1(\mu)$ if they are equal $\mu$-a.e.)  for which their absolute value is Lebesgue integrable (Proposition~\ref{prop:GI}).
This geometric characterization allows us to design a generic dichotomic search algorithm (Algorithm~1) to approximate the Chernoff optimal skewing parameter, generalizing the prior work~\cite{CI-2013}. 
When $P$ and $Q$ belong to a same exponential family, we recover in \S\ref{sec:ChernoffIG} the results of~\cite{CI-2013}.
This geometric characterization also allows us to reinterpret the Chernoff information as a minmax symmetrization 
of the Kullback--Leibler divergence, and we define by analogy the forward and reverse Chernoff--Bregman divergences in \S\ref{sec:CBD} (Definition~\ref{def:CBD}).
In \S\ref{sec:CIGauss}, we consider the Chernoff information between Gaussian distributions:
We show that the optimality condition for the  Chernoff information between univariate Gaussian distributions can be solved exactly and report a closed-form formula for the Chernoff information between any two univariate Gaussian distributions (Proposition~\ref{prop:unigauss}).
For multivariate Gaussian distributions, we show how to implement the dichotomic search algorithms to approximate the Chernoff information, and report a closed-form formula for the Chernoff information between two centered multivariate Gaussian distributions with scaled covariance matrices (Proposition~\ref{prop:CIscaleCov}).
Finally, we conclude in~\S\ref{sec:concl}

\section{Chernoff information from the viewpoint of likelihood ratio exponential families}\label{sec:CILREF}

\begin{figure}[h]
\centering
\includegraphics[width=0.9\textwidth]{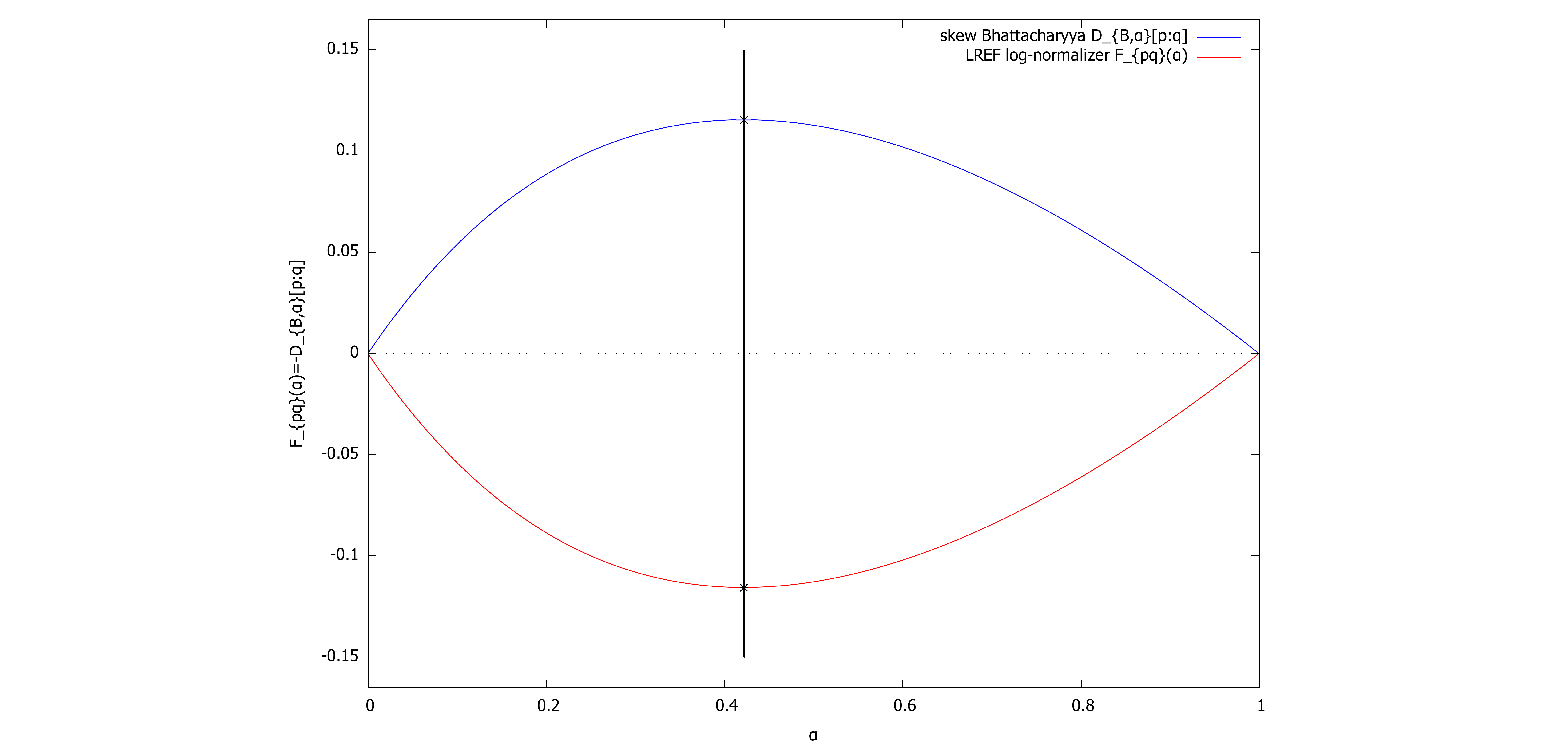}
\caption{Plot of the Bhattacharryya distance $D_{B,\alpha}(p:q)$ (strictly concave, displayed in blue) and the log-normalizer $F_{pq}(\alpha)$ of the induced LREF $\calE_{pq}$ (strictly convex, displayed in red) for two univariate normal densities $p=p_{0,1})$ (standard normal) and $q=p_{1,2}$: The curves $D_{B,\alpha}(p:q)=-F_{pq}(\alpha)$ are horizontally mirror symmetric to each others. The Chernoff information optimal skewing   value $\alphastar$ between these two univariate normal distributions can be calculated exactly in closed-form, 
see \protect\S\ref{sec:ChernoffUniGaussian} 
(approximated numerically here for plotting the vertical grey line by $\alphastar\approx 0.4215580558605244$).}
\label{fig:FBhatEx1}
\end{figure}

\subsection{LREFs and the Chernoff information}

Recall that $L^1(\mu)$ denotes the Lebesgue vector space of measurable functions $f$ such that $\int_\calX |f|\dmu<\infty$. 
Given two prescribed densities $p$ and $q$ of $L^1(\mu)$, consider building a uniparametric  exponential family~\cite{barndorff2014information} $\calE_{pq}$ which consists of the weighted geometric mixtures of $p$ and $q$:
$$
\calE_{pq}\eqdef \left\{ (pq)^G_\alpha(x)\eqdef \frac{p(x)^\alpha q(x)^{1-\alpha}}{Z_{pq}(\alpha)} \st \alpha\in\Theta\right\},
$$
where 
$$
Z_{pq}(\alpha)=\int_{\calX} p(x)^\alpha q(x)^{1-\alpha}\dmu(x)=\rho_\alpha[p:q]
$$ 
denotes the normalizer (or partition function) of the geometric mixture 
$$
(pq)^G_\alpha(x)\propto p(x)^\alpha q(x)^{1-\alpha}
$$ 
so that $\int_\calX (pq)^G_\alpha\dmu=1$ (Figure~\ref{fig:exparc}). 

Let us express the density $(pq)^G_\alpha$ in the canonical form ($\ast$) of exponential families~\cite{barndorff2014information}:
\begin{eqnarray*}
(pq)^G_\alpha(x)&=& \exp\left(\alpha\log\frac{p(x)}{q(x)}-\log Z_{pq}(\alpha)\right) \, q(x),\\
&\stackrel{\ast}{=:}& \exp\left(\alpha t(x)-F_{pq}(\alpha)+k(x)\right).
\end{eqnarray*}
It follows from this decomposition that
 $\alpha\in\Theta\subset\bbR$ is the scalar natural parameter, $t(x)=\log\frac{p(x)}{q(x)}$ denotes the sufficient statistic (minimal when $p(x)\not= q(x)$ $\mu$-a.e.), $k(x)=\log q(x)$ is an auxiliary carrier term wrt. measure $\mu$ (i.e., measure $\mathrm{d}\nu(x)=q(x)\dmu(x)$),  and 
$$
F_{pq}(\alpha)=\log Z_{pq}(\alpha)=-D_{B,\alpha}[p:q]<0
$$
 is the log-normalizer (or log-partition or cumulant function). 
Since the sufficient statistic is the logarithm of the likelihood ratio of $p(x)$ and $q(x)$, Gr\"unwald~\cite{grunwald2007minimum} (Sec. 19.6) termed $\calE_{pq}$ a  {\em Likelihood Ratio Exponential Family} (LREF). See also~\cite{brekelmans2020lref} for applications of LREFs to  Markov chain Monte Carlo (McMC) methods.

We have $p=(pq)^G_1$ and $q=(pq)^G_0$.
Thus let $\alpha_p=1$ and  $\alpha_q=0$, and let us interpret geometrically $\{(pq)^G_\alpha, \alpha\in \Theta\}$ as a maximal exponential arc~\cite{cena2007exponential,de2018mixture,siri2019minimization} where $\Theta\subseteq \bbR$ is an interval.
We denote by $\calE_{\overline{pq}}$ the open exponential arc with extremities $p$ and $q$.

Since the log-normalizers $F(\theta)$ of exponential families are always strictly convex and real analytic~\cite{barndorff2014information} (i.e., $F(\theta)\in C^\omega(\bbR)$), we deduce that $D_{B,\alpha}[p:q]=-F_{pq}(\alpha)$ is strictly concave and real analytic. Moreover, we have $D_{B,0}[p:q]=D_{B,1}[p:q]=0$. Hence, the Chernoff  optimal skewing parameter  $\alphastar$ is {\em unique} when $p\not=q$ $\mu$-a.e., 
and we get the Chernoff information calculated as  
$$
D_C[p:q]=D_{B,\alphastar}(p,q).
$$
See Figure~\ref{fig:FBhatEx1} for a plot of the strictly concave function $D_{B,\alpha}[p:q]$ and the strictly convex function $F_{pq}(\alpha)=-D_{B,\alpha}[p:q]$ when $p=p_{0,1}$ is the standard normal density and $q=p_{1,2}$ is a normal density of mean $1$ and variance $2$ where $$
p_{\mu,\sigma^2}(x):=\frac{1}{\sqrt{2\pi\sigma^2}} \exp\left(-\frac{1}{2}\frac{(x-\mu)^2}{\sigma^2}\right).
$$

Consider the {\em full} natural parameter space $\Theta_{pq}$ of $\calE_{pq}$:
$$
\Theta_{pq}=\left\{ \alpha\in\bbR \st \rho_\alpha(p:q)<+\infty \right\}.
$$
The natural parameter space $\Theta_{pq}$  is always convex~\cite{barndorff2014information} and since $\rho_0(p:q)=\rho_1(p:q)=1$, we necessarily have $(0,1)\in \Theta_{pq}$ but not necessarily $[0,1]\in\Theta_{pq}$ as detailed in the following remark: 

\begin{Remark}
In order to be an exponential family, the densities $(pq)^G_\alpha$ shall have the same coinciding support for all values of $\alpha$ belonging to the natural parameter space.
The support of the geometric mixture density $(pq)^G_\alpha$ is
$$
\supp\left((pq)^G_\alpha\right)=\left\{
\begin{array}{ll}
\supp(p)\cap\supp(q), & \alpha\in\Theta_{pq}\backslash\{0,1\}\\
\supp(p), & \alpha=1\\
\supp(q), & \alpha=0.\\
\end{array}
\right.
$$
This condition is trivially satisfied when the supports of $p$ and $q$ coincide, and therefore $[0,1]\subset\Theta_{pq}$ in that case.
Otherwise, we may consider the common support $\calX_{pq}=\supp(p)\cap\supp(q)$ for $\alpha\in (0,1)$.
In this latter case, we are poised to restrict the natural parameter space to $\Theta_{pq}=(0,1)$ even if $\rho_\alpha(p:q)<\infty$ for some $\alpha$ outside that range.
\end{Remark}

To emphasize that $\alphastar$ depends on $p$ and $q$, we shall use the notation $\alphastar(p:q)$ whenever necessary.
We have $\alphastar(q:p)=1-\alphastar(p:q)$, and since $D_{B,\alpha}(p:q)=D_{B,1-\alpha}(q:p)$, and we check that
$$
D_C[p,q]=D_{B,\alphastar(p:q)}(p:q)=D_{B,\alphastar(q:p)}(q:p)=D_C[q,p].
$$
Thus the skewing value $\alphastar(q:p)$ may be called the conjugate Chernoff exponent (i.e., depends on the convention chosen for interpolating on the exponential arc).

However, since the Chernoff information does not satisfy the triangle inequality, it is not a metric distance and the Chernoff information is called a quasi-distance.

\begin{Proposition}[Uniqueness of the Chernoff information optimal skewing parameter]\label{prop:CIalphaunique}
Let $P$ and $Q$ be two probability measures dominated by a positive measure $\mu$ with corresponding Radon-Nikodym densities $p$ and $q$, respectively.
The Chernoff information optimal skewing parameter $\alphastar(p:q)$ is unique when $p\not=q$ $\mu$-almost everywhere, and
$$
D_C[p,q]=D_{B,\alphastar(p:q)}(p:q)=D_{B,\alphastar(q:p)}(q:p)=D_C[q,p].
$$  
When $p=q$ $\mu$-a.e., we have $D_C[p:q]=0$ and $\alphastar$ is undefined since it can range in $[0,1]$.
\end{Proposition}

\begin{Definition}
An exponential family is called regular~\cite{barndorff2014information} when the natural parameter space $\Theta$ is open, i.e., $\Theta=\Theta^\circ$ where $\Theta^\circ$ denotes the interior of $\Theta$ (i.e., an open interval).
\end{Definition}

\begin{Proposition}[Finite sided Kullback-Leibler divergences]\label{eq:KLDfinite}
When the LREF $\calE_{pq}$ is a regular exponential family with natural parameter space $\Theta\supsetneq [0,1]$, 
both the forward Kullback-Leibler divergence $D_\KL[p:q]$ and 
the reverse Kullback-Leibler divergence $D_\KL[q:p]$ are finite.
\end{Proposition}

\begin{proof}
A reverse divergence $D^*(\theta_1:\theta_2)$ is a divergence on the swapped parameter order: $D^*(\theta_1:\theta_2):=D(\theta_2:\theta_1)$.
We shall use the result pioneered in~\cite{azoury2001relative,collins2001generalization} that the KLD between two densities 
$p_{\theta_1}$ and $p_{\theta_2}$
of a regular exponential family $\calE=\{p_\theta\st\theta\in\Theta\}$ amounts to a reverse Bregman divergence $(B_F)^*$ (i.e., a Bregman divergence on swapped parameter order) induced by the log-normalizer of the family:
$$
D_\KL[p_{\theta_1}:p_{\theta_2}]=(B_F)^*(\theta_1:\theta_2)=B_F(\theta_2:\theta_1),
$$
where $B_F$ is the Bregman divergence defined on domain $D=\dom(F)$ (see Definition~1 of~\cite{banerjee2005clustering}):
\begin{eqnarray*}
B_F: && D\times \ri(D) \rightarrow [0,\infty)\\
&& (\theta_1,\theta_2) \mapsto  B_F(\theta_1:\theta_2)=F(\theta_1)-F(\theta_2)-(\theta_1-\theta_2)^\top \nabla F(\theta_2)<+\infty,
\end{eqnarray*}
where $\ri(D)$ denotes the relative interior of domain $D$. Bregman divergences are always finite and the only symmetric Bregman divergences are squared Mahalanobis distances~\cite{nielsen2009sided} (i.e., with corresponding Bregman generators defining quadratic forms).

For completeness, we recall the proof   as follows: We have 
$$
\log\frac{p_{\theta_1}(x)}{p_{\theta_2}(x)}=(\theta_1-\theta_2)^\top t(x)-F(\theta_1)+F(\theta_2).
$$ 
Thus we get
\begin{eqnarray*}
D_\KL[p_{\theta_1}:p_{\theta_2}]&=&E_{p_{\theta_1}}\left[\log\frac{p_{\theta_1}}{p_{\theta_2}}\right],\\
&=& F(\theta_2)-F(\theta_1)-(\theta_1-\theta_2)^\top E_{p_{\theta_1}}[t(x)],
\end{eqnarray*}
using the linearity property of the expectation operator.
When $\calE$ is regular, we also have $E_{p_{\theta_1}}[t(x)]=\nabla F(\theta_1)$ (see~\cite{sundberg2019statistical}), and therefore we get
$$
D_\KL[p_{\theta_1}:p_{\theta_2}]=F(\theta_2)-F(\theta_1)-(\theta_1-\theta_2)^\top  \nabla F(\theta_1)=:B_F(\theta_2:\theta_1)=(B_F)^*(\theta_1:\theta_2).
$$

In our LREF setting, we thus have: 
$$
D_\KL[p:q]=(B_F)^*(\alpha_p:\alpha_q)=B_{F_{pq}}(\alpha_q:\alpha_p)=B_{F_{pq}}(0:1),
$$ 
and $D_\KL[q:p]=B_{F_{pq}}(\alpha_p:\alpha_q)=B_{F_{pq}}(1:0)$ where $B_{F_{pq}}(\alpha_1:\alpha_2)$ denotes the following scalar Bregman divergence:
$$
B_{F_{pq}}(\alpha_1:\alpha_2)=F_{pq}(\alpha_1)-F_{pq}(\alpha_2)-(\alpha_1-\alpha_2)F_{pq}'(\alpha_2).
$$
Since $F_{pq}(0)=F_{pq}(1)=0$ and $B_{F_{pq}}: \Theta\times \ri(\Theta) \rightarrow [0,\infty)$, we have 
$$
D_\KL[p:q]=B_{F_{pq}}(\alpha_q:\alpha_p)=B_{F_{pq}}(0:1)=F_{pq}'(1)<\infty. 
$$ 
Similarly
$$
D_\KL[q:p]=B_{F_{pq}}(\alpha_p:\alpha_q)=B_{F_{pq}}(1:0)=-F_{pq}'(0)<\infty.
$$
Notice that since $B_{F_{pq}}(\alpha_1:\alpha_2)> 0$, we have  
$F_{pq}'(1)>0$ and $F_{pq}'(0)<0$ when $p\not=q$ $\mu$-almost everywhere. 
Moreover, since $F_{pq}(\alpha)$ is strictly convex, $F_{pq}'(\alpha)$ is strictly monotonically increasing, and  therefore there exists a unique $\alphastar\in(0,1)$ such that $F_{pq}'(\alphastar)=0$.
\end{proof}

\begin{Example}
When $p$ and $q$ belongs to a same regular exponential family $\calE$ (e.g., $p$ and $q$ are two normal densities), their sided KLDs~\cite{nielsen2009sided} are both finite. The LREF induced by two Cauchy distributions $p_{l_1,s_1}$ and $p_{l_2,s_2}$ 
is such that $[0,1]\subset \Theta$ since the skewed Bhattacharyya distance is defined and finite for $\alpha\in\bbR$~\cite{nielsen2021f}. 
Therefore the KLDs between two Cauchy distributions are always finite~\cite{nielsen2021f}, see the closed-form formula in~\cite{chyzak2019closed}.
\end{Example}

\begin{Remark}
If $0\in \Theta^\circ$, then $B_{F_{pq}}(1:0)<\infty$ and therefore $D_\KL[q:p]<\infty$.
Since the KLD between a standard Cauchy distribution $p$ and a standard normal distribution $q$ is $+\infty$, 
we deduce that $D_\KL[p:q]\not= B_{F_{pq}}(0:1)$, and therefore $1\not\in\Theta^\circ$.
Similarly, when $1\in\Theta^\circ$, we have $B_{F_{pq}}(0:1)<\infty$ and therefore $D_\KL[p:q]<\infty$.
\end{Remark}

\begin{Proposition}[Chernoff information expressed as KLDs]\label{prop:COC}
We have at the Chernoff information optimal skewing value $\alphastar\in(0,1)$ the following identities:
$$
D_C[p:q]=D_\KL[(pq)^G_\alphastar:p]=D_\KL[(pq)^G_\alphastar:q].
$$
\end{Proposition}

\begin{figure}
\centering
\includegraphics[width=10cm]{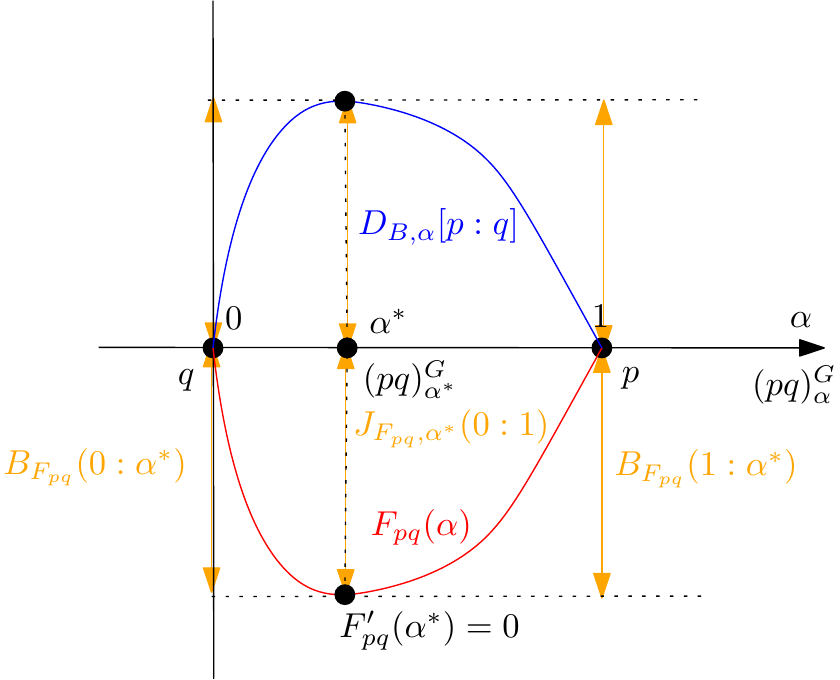}

\caption{The best unique parameter $\alpha^*$ defining the Chernoff information optimal skewing parameter is found by setting the derivative of the strictly convex function $F_{pq}(\alpha)$ to zero.
At the optimal value $\alphastar$, we have $D_C[p:q]=D_\KL[(pq)^G_\alphastar:p]=D_\KL[(pq)^G_\alphastar:q]=-F(\alphastar)>0$.}\label{fig:chernoffexponent}
\end{figure}

\begin{proof}
Since the skewed Bhattacharyya distance between two densities $p_{\theta_1}$ and $p_{\theta_2}$ of an exponential family with log-normalizer $F$ amounts to a skew Jensen divergence for the log-normalizer~\cite{BhatEF-Huzurbazar-1955,nielsen2011burbea}, we have:
$$
D_{B,\alpha}(p_{\theta_1}:p_{\theta_2})=J_{F,\alpha}(\theta_1:\theta_2),
$$
where the skew Jensen divergence~\cite{burbea1982convexity} is given by
$$
J_{F,\alpha}(\theta_1:\theta_2)=\alpha F(\theta_1)+(1-\alpha) F(\theta_2)-F(\alpha\theta_1+(1-\alpha)\theta_2).
$$

In the setting of the LREF, we have
\begin{eqnarray*}
D_{B,\alpha}((pq)_{\alpha_1}^G:(pq)_{\alpha_2}^G) &=& J_{F_{pq},\alpha}(\alpha_1:\alpha_2),\\
&=& \alpha F_{pq}(\alpha_1)+(1-\alpha) F_{pq}(\alpha_2)-F_{pq}(\alpha\alpha_1+(1-\alpha)\alpha_2).
\end{eqnarray*}

At the optimal value $\alphastar$, we have $F_{pq}'(\alphastar)=0$.
Since $D_\KL[(pq)^G_\alphastar:p]=B_{F_{pq}}(1:\alphastar)=-F(\alphastar)$ and $D_\KL[(pq)^G_\alphastar:q]=B_{F_{pq}}(0:\alphastar)=-F(\alphastar)$ and $D_C[p:q]=-\log\rho_\alphastar(p:q)=J_{F_{pq},\alphastar}(1:0)=-F_{pq}(\alphastar)$, we get
$$
D_C[p:q]=D_\KL[(pq)^G_\alphastar:p]=D_\KL[(pq)^G_\alphastar:q].
$$
Figure~\ref{fig:chernoffexponent} illustrates the proposition on the plot of the scalar function $F_{pq}(\alpha)$.
\end{proof}

\begin{Corollary}
The Chernoff information optimal skewing value $\alphastar(p:q)\in(0,1)$ can be used to calculate the Chernoff information $D_C[p,q]$ as a Bregman divergence induced by the LREF:
$$
D_C[p:q]=B_{F_{pq}}[1:\alphastar]= B_{F_{pq}}[0:\alphastar].
$$
\end{Corollary}

Proposition~\ref{prop:COC} let us interpret the Chernoff information as a special  symmetrization of the Kullback--Leibler divergence~\cite{chen2008metrics2}, different from the Jeffreys divergence or the Jensen-Shannon divergence~\cite{nielsen2019jensen}.
Indeed, the Chernoff information can be rewritten as 
\begin{equation}
D_C[p:q] = \min_{r\in\calE_{\overline{pq}}}\{D_\KL[r:p],D_\KL[r:q]\}.
\end{equation}
As such, we can interpret the Chernoff information as the radius of a minimum enclosing left-sided Kullback--Leibler ball on the space $L^1(\mu)$.

\begin{figure}
\centering
\includegraphics[width=11cm]{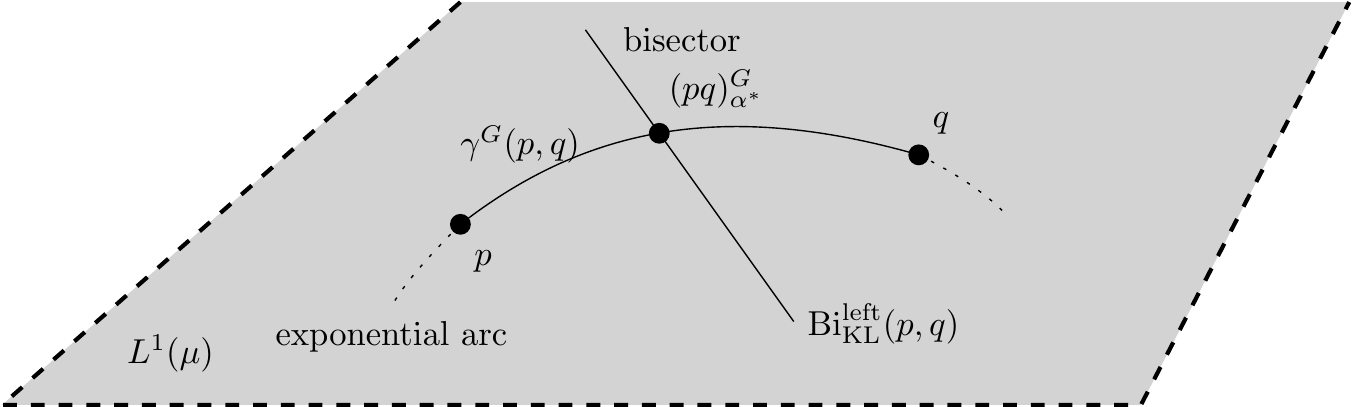}
\caption{The Chernoff information distribution $(PQ)_\alphastar^G$ with density $(pq)_\alphastar^G$ is obtained as the unique intersection of the exponential arc $\gamma^G(p,q)$ linking density $p$ to density $q$ of $L^1(\mu)$ with the left-sided Kullback-Leibler divergence bisector $\mathrm{Bi}_\KL^{\mathrm{left}}(p,q)$ of $p$ and $q$: $(pq)^G_\alphastar=\gamma^G(p,q) \cap\Bi_\KL^\mathrm{left}(p,q)$.}\label{fig:exparc}
\end{figure}

\subsection{Geometric characterization of the Chernoff information and the Chernoff information distribution}

Let us term the  probability distribution $(PQ)_{\alphastar}^G\ll \mu$ with corresponding density $(pq)_{\alphastar}^G$ the 
{\em Chernoff information distribution} to avoid confusion with another concept of Chernoff distributions~\cite{han2022berry} used in statistics.
We can characterize geometrically the Chernoff information distribution $(pq)^G_\alphastar$ on $L^1(\mu)$ 
as the intersection of a {\em left-sided Kullback-Leibler divergence bisector}:
$$
\Bi_\KL^\mathrm{left}(p,q)\eqdef \left\{ r\in L^1(\mu)\st D_\KL[r:p]=D_\KL[r:q]\right\},
$$
with an {\em exponential arc}~\cite{cena2007exponential}
$$
\gamma^G(p,q)\eqdef \left\{(pq)^G_\alpha \st \alpha\in [0,1]\right\}.
$$

We thus interpret Proposition~\ref{prop:COC} geometrically by the following proposition (see Figure~\ref{fig:exparc}):

\begin{Proposition}[Geometric characterization of the Chernoff information]\label{prop:GI}
On the vector space $L^1(\mu)$, the Chernoff information distribution is the unique distribution
$$
(pq)^G_\alphastar=\gamma^G(p,q) \cap\Bi_\KL^\mathrm{left}(p,q).
$$
\end{Proposition}
The point $(pq)^G_\alphastar$ has been called the {\em Chernoff point} in~\cite{CI-2013}.

Proposition~\ref{prop:GI} allows us to design a  dichotomic search to numerically approximate $\alphastar$ as 
reported in pseudo-code in Algorithm~1 (see also the illustration in Figure~\ref{fig:Dicho}).

\noindent (Algorithm~1). Dichotomic search for approximating the Chernoff information by approximating the optimal skewing parameter value $\tilde\alpha\approx\alphastar$ and reporting $D_C[p:q]\approx D_\KL[(pq)^G_{\tilde\alpha}:p]$. The search requires $\ceil{\log_2\frac{1}{\epsilon}}$ iterations to guarantee $|\alphastar-\tilde\alpha|\leq \epsilon$.

 \begin{algorithm} 
                \SetAlgoLined
                \SetKwInOut{Input}{input}
                \SetKwInOut{Ret}{return}
                \Input{Two densities $p$, $q$ of $L^1(\mu)$, and a numerical precision threshold $\epsilon>0$}
								$\alpha_m=0$\; $\alpha_M=1$\;
                    \While{$|\alpha_M-\alpha_m|>\epsilon$}{
                    $\alpha=\frac{\alpha_m+\alpha_M}{2}$\;
            \If{$D_\KL[(pq)^G_\alpha:p]>D_\KL[(pq)^G_\alpha:q]$}{$\alpha_m=\alpha$\; \tcp{See Figure~\ref{fig:Dicho} for an illustration and Proposition~\ref{prop:GI}}}
						\Else{$\alpha_M=\alpha$\;}
															} 
                    \Return{ 
                     $D_\KL[(pq)^G_\alpha:p]$\;
                    }
            \end{algorithm}

\begin{figure}
\centering
\includegraphics[width=0.38\textwidth]{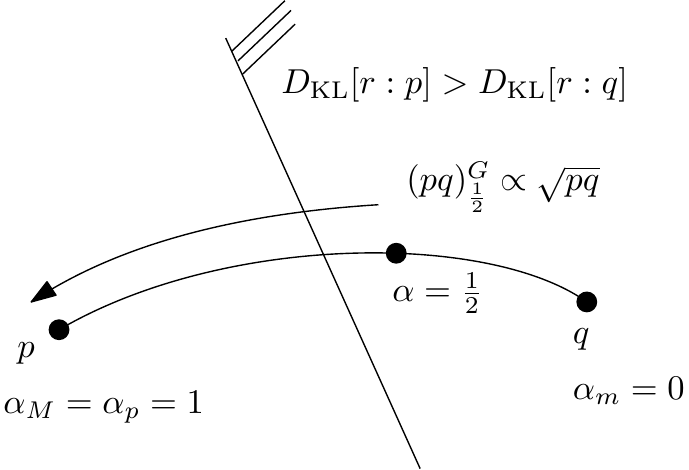}

\caption{Illustration of the dichotomic search for approximating the optimal skewing parameter $\alphastar$ to within some prescribed numerical precision $\epsilon>0$.
\label{fig:Dicho}}
\end{figure}

\begin{Remark}
We do not need to necessarily handle normalized densities $p$ and $q$ since we have for $\alpha\in\bbR\backslash\{0,1\}$:
$$
(pq)^G_\alpha=(\tilde{p}\tilde{q})^G_\alpha,
$$
where $p(x)=\frac{\tilde{p}(x)}{Z_p}$ and $q(x)=\frac{\tilde{q}(x)}{Z_q}$ with $\tilde{p}$ and $\tilde{q}$ denoting 
the computationally-friendly unnormalized positive densities.
This property of geometric mixtures is used in Annealed Importance Sampling~\cite{neal2001annealed,grosse2013annealing} (AIS), and
for designing an asymptotically efficient estimator for computationally-intractable parametric densities~\cite{takenouchi2019parameter} $\tilde{q}_\theta$ (e.g., distributions learned by Boltzmann machines).
\end{Remark}

\subsection{Dual parameterization of LREFs}
The densities $(pq)_\alpha^G$ of a LREF can also be parameterized by their dual moment parameter~\cite{barndorff2014information} (or mean parameter): 
$$
\beta=\beta(\alpha):=E_{(pq)_\alpha^G}[t(x)]=E_{(pq)_\alpha^G}\left[\log\frac{p(x)}{q(x)}\right].
$$
When the LREF is regular (and therefore steep~\cite{sundberg2019statistical}), we have $\beta(\alpha)=F_{pq}'(\alpha)$ and 
$\alpha={F_{pq}^*}'(\beta)$, where $F_{pq}^*$ denotes the Legendre transform of $F_{pq}$.
At the optimal value $\alphastar$, we have $F'_{pq}(\alphastar)=0$. 
Therefore an equivalent condition of optimality is
$$
\beta(\alphastar)=F_{pq}'(\alphastar)=0=E_{(pq)_\alphastar^G}\left[\log\frac{p(x)}{q(x)}\right].
$$

Notice that when $[0,1]\subset \Theta^\circ$, we have finite forward and reverse Kullback--Leibler divergences:
\begin{itemize}
\item $\alpha=1$, we have $(pq)_1^G=p$ and 
$$
\beta(1)=E_p\left[\log\frac{p(x)}{q(x)}\right]=D_\KL[p:q]=F_{pq}'(1)>0.
$$
\item $\alpha=0$, we have $(pq)_0^G=q$ and
$$
\beta(0)=E_q\left[\log\frac{p(x)}{q(x)}\right]=-D_\KL[q:p]=F_{pq}'(0)<0.
$$
\end{itemize} 

Since $F_{pq}(\alpha)$ is strictly convex, we have $F_{pq}''(\alpha)>0$ 
and $F_{pq}'$ is strictly increasing with $F_{pq}'(0)=-D_\KL[q:p]<0$ and
$F_{pq}'(1)=D_\KL[p:q]>0$. The value $\alphastar$ is thus the unique value such that $F_{pq}'(\alphastar)=0$.

\begin{Proposition}[Dual optimality condition for the Chernoff information]\label{prop:equivcond}
The unique Chernoff information optimal skewing parameter $\alphastar$ is such that
$$
\OC_\alpha: D_\KL[(pq)^G_\alphastar:p]=D_\KL[(pq)^G_\alphastar:q] \Leftrightarrow \OC_\beta: \beta(\alphastar)=E_{(pq)_\alphastar^G}\left[\log\frac{p(x)}{q(x)}\right]=0.
$$
\end{Proposition}

As a side remark, let us notice that the Fisher information of a likelihood ratio exponential family $\calE_{pq}$ is 
$$
I_{pq}(\alpha)=-E_{(pq)_\alpha^G}[(\log (pq)_\alpha^G)'']=F''_{pq}(\alpha)>0,
$$
and $=F''_{pq}(\alpha){F^*}''_{pq}(\beta)=1$.
 
\section{Chernoff information between densities of an exponential family}\label{sec:ChernoffIG}

\subsection{General case}

We shall now consider that the densities $p$ and $q$ (with respect to measure $\mu$) belong to a same exponential family~\cite{barndorff2014information}:
$$
\calE=\left\{ P_\lambda \st \frac{\dP_\lambda}{\dmu}=p_\lambda(x)=\exp(\theta(\lambda)^\top t(x)-F(\theta(\lambda))),\quad \lambda\in\Lambda\right\},
$$
where $\theta(\lambda)$ denotes the natural parameter associated with the ordinary parameter $\lambda$, $t(x)$ the sufficient statistic vector and $F(\theta(\lambda))$ the log-normalizer.
When $\theta(\lambda)=\lambda$ and $t(x)=x$, the exponential family is called a natural exponential family (NEF).
The exponential family $\calE$ is defined by $\mu$ and $t(x)$, hence we may write when necessary $\calE=\calE_{\mu,t}$.

\begin{Example}
The set of univariate Gaussian distributions 
$$
\mathcal{N}=\{p_{\mu,\sigma^2}(x) \st \lambda=(\mu,\sigma^2) \in \Lambda=\bbR \times \bbR_{++}\}
$$ 
forms an exponential family with the following decomposition terms:
\begin{eqnarray*}
\lambda &=& (\mu,\sigma^2)\in\Lambda=\bbR\times\bbR_{++},\\
\theta(\lambda) &=& \left(\theta_1=\frac{\mu}{\sigma^2},\theta_2=-\frac{1}{2\sigma^2}\right)\in \Theta=\bbR\times\bbR_{--},\\
t(x) &=& (x,x^2),\\
F(\theta)&=&-\frac{\theta_1^2}{4\theta_2}+\frac{1}{2}\log\left(-\frac{\pi}{\theta_2}\right),
\end{eqnarray*}
where $\bbR_{++}=\{x\in\bbR \st x>0\}$ and $\bbR_{--}=\{x\in\bbR \st x<0\}$ denotes the set of positive real numbers and negative real numbers, respectively.
Letting $v=\sigma^2$ be the variance parameter, we get the equivalent natural parameters $\left(\frac{\mu}{v},-\frac{1}{2v}\right)$.
The log-normalizer can be written using the $(\mu,v)$-parameterization as
$F(\mu,v)=\frac{1}{2}\log(2\pi v)+\frac{\mu^2}{2v}$ and $\theta=\left(\theta_1=\frac{\mu}{v},-\frac{1}{2v}\right)$.
See Appendix~\ref{sec:unigaussian} for further details concerning this normal exponential family.
\end{Example}

Notice that we can check easily that the LREF between two densities of an exponential family
forms a 1D sub-exponential family of the exponential family:
\begin{eqnarray*}
p_{\theta_1}(x)^\alpha\, p_{\theta_2}(x)^{1-\alpha} &\propto& \exp(\inner{\alpha\theta_1+(1-\alpha)\theta_2}{t(x)}-\alpha F(\theta_1)-(1-\alpha)F(\theta_2)),\\
&=&  p_{\alpha\theta_1+(1-\alpha)\theta_2}(x) \exp(F(\alpha\theta_1+(1-\alpha)\theta_2)-\alpha F(\theta_1)-(1-\alpha)F(\theta_2))),\\
&=& p_{\alpha\theta_1+(1-\alpha)\theta_2}(x)  \exp(-J_{F,\alpha}(\theta_1:\theta_2)),
\end{eqnarray*}
where $J_F$ denote the Jensen divergence induced by $F$.
 
The optimal skewing value condition of the Chernoff information between two categorical distributions~\cite{cover1999elements} was extended  to densities $p_{\theta_1}$ and $p_{\theta_2}$ of an exponential family in~\cite{CI-2013}. 
The family of categorical distributions with $d$ choices forms an exponential family with natural parameter of dimension $d-1$.
Thus Proposition~\ref{prop:cief} generalizes the analysis in~\cite{cover1999elements}.

Let $p=p_{\theta_1}$ and $q=p_{\theta_2}$. Then we have the property that exponential families are closed under geometric mixtures:
$$
(p_{\theta_1}p_{\theta_2})_\alpha^G = p_{\alpha\theta_1+(1-\alpha)\theta_2}.
$$
Since the natural parameter space $\Theta$ is convex, we have $\alpha\theta_1+(1-\alpha)\theta_2\in\Theta$.

The KLD between two densities $p_{\theta_1}$ and $p_{\theta_2}$ of a regular exponential family $\calE$ amounts to a reverse Bregman divergence for the log-normalizer of $\calE$:
$$
D_\KL[p_{\theta_1}:p_{\theta_2}]=B_F(\theta_2:\theta_1),
$$
where $B_F(\theta_2:\theta_1)$ denotes the Bregman divergence:
$$
B_F(\theta_2:\theta_1)=F(\theta_2)-F(\theta_1)-(\theta_2-\theta_1)^\top \nabla F(\theta_1).
$$

Thus when the exponential family $\calE$ is regular, both the forward and reverse KLD are finite,
 and we can rewrite Proposition~\ref{prop:COC} to characterize $\alphastar$ as follows:
\begin{equation}\label{eq:ocbd}
\OC_\EF:\quad B_F(\theta_1:\theta_\alphastar)=B_F(\theta_2:\theta_\alphastar),
\end{equation}
where $\theta_\alphastar=\alphastar\theta_1+(1-\alphastar)\theta_2$.

The Legendre-Fenchel transform of $F(\theta)$ yields the convex  conjugate 
$$
F^*(\eta)=\sup_{\theta\in\Theta} \{\theta^\top\eta-F(\theta)\}
$$ 
with 
$\eta(\theta)=\nabla F(\theta)$. 
Let $H=\{\eta(\theta) \st \theta\in\Theta\}$ denote the dual moment parameter space also called domain of means.
The Legendre transform associates to $(\Theta,F(\theta))$ the convex conjugate $(H,F^*(\eta))$.
In order for $(H,F^*(\eta))$ to be of the same well-behaved type of $(\Theta,F(\theta))$, we shall consider convex functions $F(\theta)$ which are steep, meaning that their gradient diverges when nearing the boundary $\bd(\Theta)$~\cite{LegendreType-1967} and thus ensures that domain $H$ is also convex. Steep convex functions are said of Legendre-type, and $((\Theta,F(\theta))^*)^*=(\Theta,F(\theta))$ (Moreau biconjugation theorem which shows that the Legendre transform is involutive).
For Legendre-type functions, there is a one-to-one mapping between parameters $\theta(\eta)$ and parameters $\eta(\theta)$ as follows:
$$
\theta(\eta)=\nabla F^*(\eta)=(\nabla F)^{-1}(\eta),
$$
and
$$
\eta(\theta)=\nabla F(\theta)=(\nabla F^*)^{-1}(\theta).
$$

Exponential families with log-normalizers of Legendre-type are called steep exponential families~\cite{barndorff2014information}.
All regular exponential families are steep, and the maximum likelihood estimator in steep exponential families exists and is unique~\cite{sundberg2019statistical} (with the likelihood equations corresponding to the method of moments for the sufficient statistics).
The set of inverse Gaussian distributions form a non-regular but steep exponential family, and the set of  
singly truncated normal distributions form a non-regular and non-steep exponential family~\cite{del1994singly} (but the exponential family of doubly truncated normal distributions is regular and hence steep).

For Legende-type convex generators $F(\theta)$, we can express the Bregman divergence $B_F(\theta_1:\theta_2)$ using the dual Bregman divergence: $B_F(\theta_1:\theta_2)=B_{F^*}(\eta_2:\eta_1)$ since there is a one-to-one correspondence between $\eta=\nabla F(\theta)$ and $\theta=\nabla F^*(\eta)$.

 For Legendre-type generators $F(\theta)$, the Bregman divergence $B_F(\theta_1:\theta_2)$ can be rewritten as the following Fenchel-Young divergence:
$$
B_F(\theta_1:\theta_2)=F(\theta_1)+F^*(\eta_2)-\theta_1^\top\eta_2:=Y_{F,F^*}(\theta_1:\eta_2).
$$

\begin{Proposition}[KLD between densities of a regular (and steep) exponential family]\label{prop:KLBD}
The KLD between two densities $p_{\theta_1}$ and $p_{\theta_2}$ of a regular and steep exponential family can be obtained equivalently as
$$
D_\KL[p_{\theta_1}:p_{\theta_2}]=B_F(\theta_2:\theta_1)=Y_{F,F^*}(\theta_2:\eta_1)=Y_{F^*,F}(\eta_1:\theta_2)=B_{F^*}(\eta_1:\eta_2),
$$
where $F(\theta)$ and its convex conjugate $F^*(\eta)$ are Legendre-type functions.
\end{Proposition}

Figure~\ref{fig:ClassificationEF} illustrates the taxonomy of regularity and steepness of exponential families by a Venn diagram. 

\begin{figure}
\centering
\includegraphics[width=0.725\textwidth]{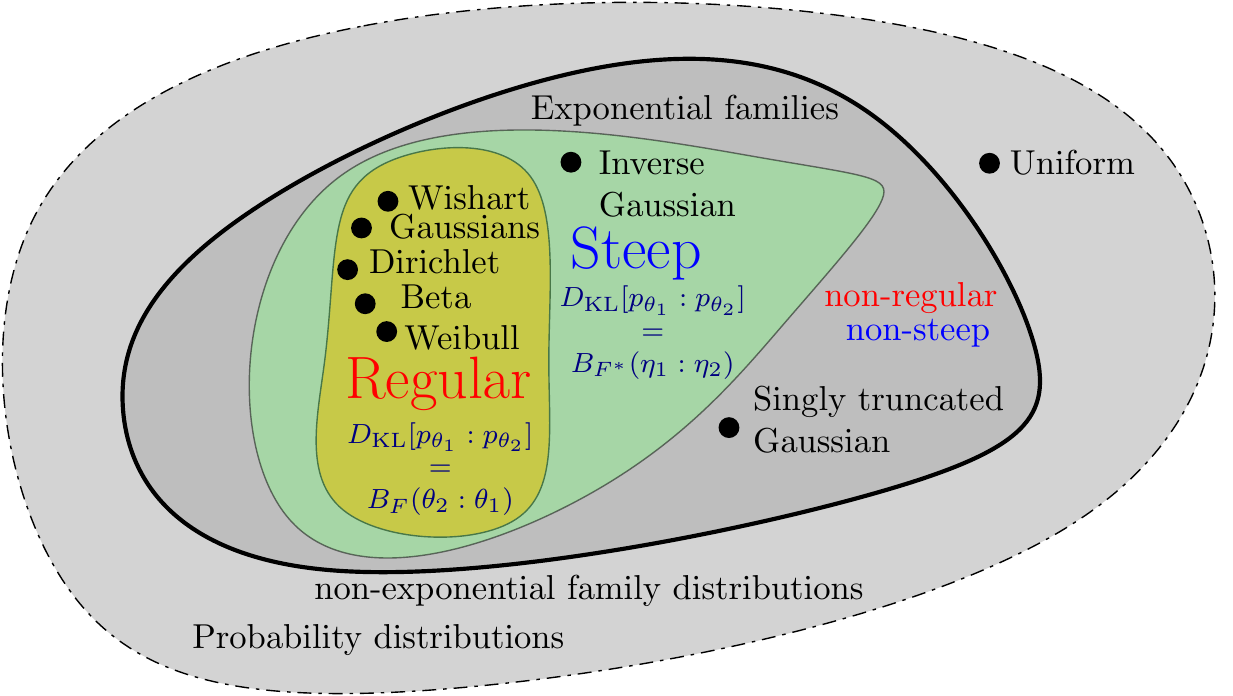}

\caption{Taxonomy of exponential families: Regular (and always steep) or steep (but not necessarily regular).
The Kullback-Leibler divergence between two densities of a regular exponential family amounts to dual Bregman divergences.
\label{fig:ClassificationEF}}
\end{figure}

It follows that the optimal condition of Eq.~\ref{eq:ocbd} can be restated as
\begin{equation}\label{eq:ocfy}
\OC_\YF:\quad Y_{F,F^*}(\theta_1:\eta_\alphastar)=Y_{F,F^*}(\theta_2:\eta_\alphastar),
\end{equation}
where $\eta_\alphastar=\nabla F^{-1}(\alphastar\theta_1+(1-\alphastar)\theta_2)$.
From the equality of Eq.~\ref{eq:ocfy}, we get the following simplified optimality condition:
\begin{equation}\label{eq:oc}
\OC_\SEF:\quad (\theta_2-\theta_1)^\top \eta_\alphastar=F(\theta_2)-F(\theta_1),
\end{equation}
where $\eta_\alphastar=\nabla F(\alphastar\theta_1+(1-\alphastar)\theta_2)$.

\begin{Remark}
We can recover ($\OC_\SEF$) by instantiating the equivalent condition 
$E_{p_{\bartheta_\alphastar}}\left[\log\frac{p_{\theta_1}}{p_{\theta_2}}\right]=0$.
Indeed, since $\log\frac{p_{\theta_1}}{p_{\theta_2}}=(\theta_1-\theta_2)^\top t(x)-F(\theta_1)+F(\theta_2)$, we get
\begin{eqnarray*}
E_{p_{\bartheta_\alphastar}}[(\theta_1-\theta_2)^\top t(x)-F(\theta_1)+F(\theta_2)] &=& 0,\\
(\theta_1-\theta_2)^\top \bareta_\alphastar=F(\theta_1)-F(\theta_2).
\end{eqnarray*}
\end{Remark}

Since the $\alpha$-skewed Bhattacharyya distance amounts to a $\alpha$-skewed Jensen divergence~\cite{nielsen2011burbea},
 we get the Chernoff information as
\begin{eqnarray*}
D_C[p_{\lambda_1}:p_{\lambda_2}] &=& J_{F,\alphastar}(\theta(\lambda_1):\theta(\lambda_2)),\\
&=& B_F(\theta_1:\theta_\alphastar) = B_F(\theta_2:\theta_\alphastar),
\end{eqnarray*}
where $J_{F,\alpha}(\theta_1:\theta_2)$ is the Jensen divergence:
$$
J_{F,\alpha}(\theta_1:\theta_2)=\alpha F(\theta_1)+(1-\alpha)F(\theta_2)-F(\alpha\theta_1+(1-\alpha)\theta_2).
$$

Notice that we have the induced LREF with log-normalizer expressed as the negative Jensen divergence induced the log-normalizer of $\calE$:
$$
F_{p_{\theta_1}p_{\theta_2}}(\alpha)=-\log \rho_\alpha[p_{\theta_1}:p_{\theta_2}]=-J_{F,\alpha}(\theta_1:\theta_2).
$$

We summarize the result in the following proposition:

\begin{Proposition}\label{prop:cief}
Let $p_{\lambda_1}$ and $p_{\lambda_2}$ be two densities of a regular exponential family $\calE$ with natural parameter $\theta(\lambda)$ and log-normalizer $F(\theta)$. Then the Chernoff information is
$$
D_C[p_{\lambda_1}:p_{\lambda_2}]=J_{F,\alphastar}(\theta(\lambda_1):\theta(\lambda_2))=B_F(\theta_1:\theta_\alphastar) = B_F(\theta_2:\theta_\alphastar),
$$ 
where $\theta_1=\theta(\lambda_1)$, $\theta_2=\theta(\lambda_2)$, and the optimal skewing parameter $\alphastar$ is unique and satisfies the following optimality condition:
\begin{equation}\label{eq:OCEF}
\OC_{\EF}:\quad (\theta_2-\theta_1)^\top \eta_\alphastar=F(\theta_2)-F(\theta_1),
\end{equation}
where $\eta_\alphastar=\nabla F(\alphastar\theta_1+(1-\alphastar)\theta_2)=E_{p_{\alphastar\theta_1+(1-\alphastar)\theta_2}}[t(x)]$.
\end{Proposition}

Figure~\ref{fig:ChernoffPoint} illustrates geometrically the Chernoff point~\cite{CI-2013} which is the geometric mixture 
$(p_{\theta_1}p_{\theta_2})_\alphastar$ induced by two comparable probability measures $P_{\theta_1},P_{\theta_2}\ll \mu$. 

In information geometry~\cite{amari2016information}, the manifold of densities $M=\{p_\theta\st\theta\in\Theta\}$ of this exponential family is a dually flat space~\cite{amari2016information} 
 $\mathcal{M}=(\{p_\theta\},g_F(\theta)=\nabla^2 F(\theta),\nabla^m,\nabla^e)$ with respect to the exponential connection $\nabla^e$ and the mixture connection $\nabla^m$, where $g_F(\theta)$ is the Fisher information metric expressed in the $\theta$-coordinate system as $\nabla^2 F(\theta)$ (and in the dual moment parameter $\eta$ as $g_F(\eta)=\nabla^2 F^*(\eta)$).
Then the exponential geodesic $\nabla^e$ is flat and corresponds to the exponential arc of geometric mixtures when parameterized with the $\nabla^e$-affine coordinate system $\theta$. 

The left-sided Kullback-Voronoi bisector: 
$$
\Bi_\KL^{\mathrm{left}}(p_{\theta_1},p_{\theta_2})=\{p_\theta \st D_\KL[p_\theta:p_{\theta_1}]=D_\KL[p_\theta:p_{\theta_1}]\}
$$
corresponds to a Bregman right-sided bisector~\cite{BVD-2010} and is $\nabla^m$ flat (i.e., an affine subspace in the $\eta$-coordinate system):
$$
\Bi_F^{\mathrm{right}}({\theta_1},{\theta_2})=\{\theta \in\Theta\st B_F({\theta_1},\theta) = B_F({\theta_2},\theta)\}.
$$

The Chernoff information distribution $(p_{\theta_1}p_{\theta_2})_\alphastar^G$ is called the Chernoff point on this exponential family manifold (see Figure~\ref{fig:ChernoffPoint}).
Since the Chernoff point is unique and since in general statistical manifolds $(\mathcal{M},g,\nabla,\nabla^*)$ can be realized by statistical models~\cite{le2006statistical}, we deduce the following proposition of interest for information geometry~\cite{amari2016information}:

\begin{Proposition}\label{prop:iguniqueinter}
Let $(\calM,g,\nabla,\nabla^*)$ be a dually flat space with corresponding canonical divergence a Bregman divergence $B_F$.
Let $\gamma^e_{pq}(\alpha)$ and $\gamma^m_{pq}(\alpha)$ be a $e$-geodesic and $m$-geodesic passing through the points $p$ and $q$ of $\calM$, respectively.
Let $\Bi^m(p,q)$ and $\Bi^e(p,q)$  be the right-sided $\nabla^m$-flat and left-sided $\nabla^e$-flat Bregman bisectors, respectively. 
Then the intersection of $\gamma^e_{pq}(\alpha)$ with $\Bi^m(p,q)$ and the intersection of $\gamma^m_{pq}(\alpha)$ with $\Bi^e(p,q)$ are unique. The point $\gamma^e_{pq}(\alpha)\cap \Bi^m(p,q)$ is called the Chernoff point and the point $\gamma^m_{pq}(\alpha)\cap\Bi^e(p,q)$ is termed the reverse or dual Chernoff point.
\end{Proposition}

\begin{figure}
\centering
\includegraphics[width=0.725\textwidth]{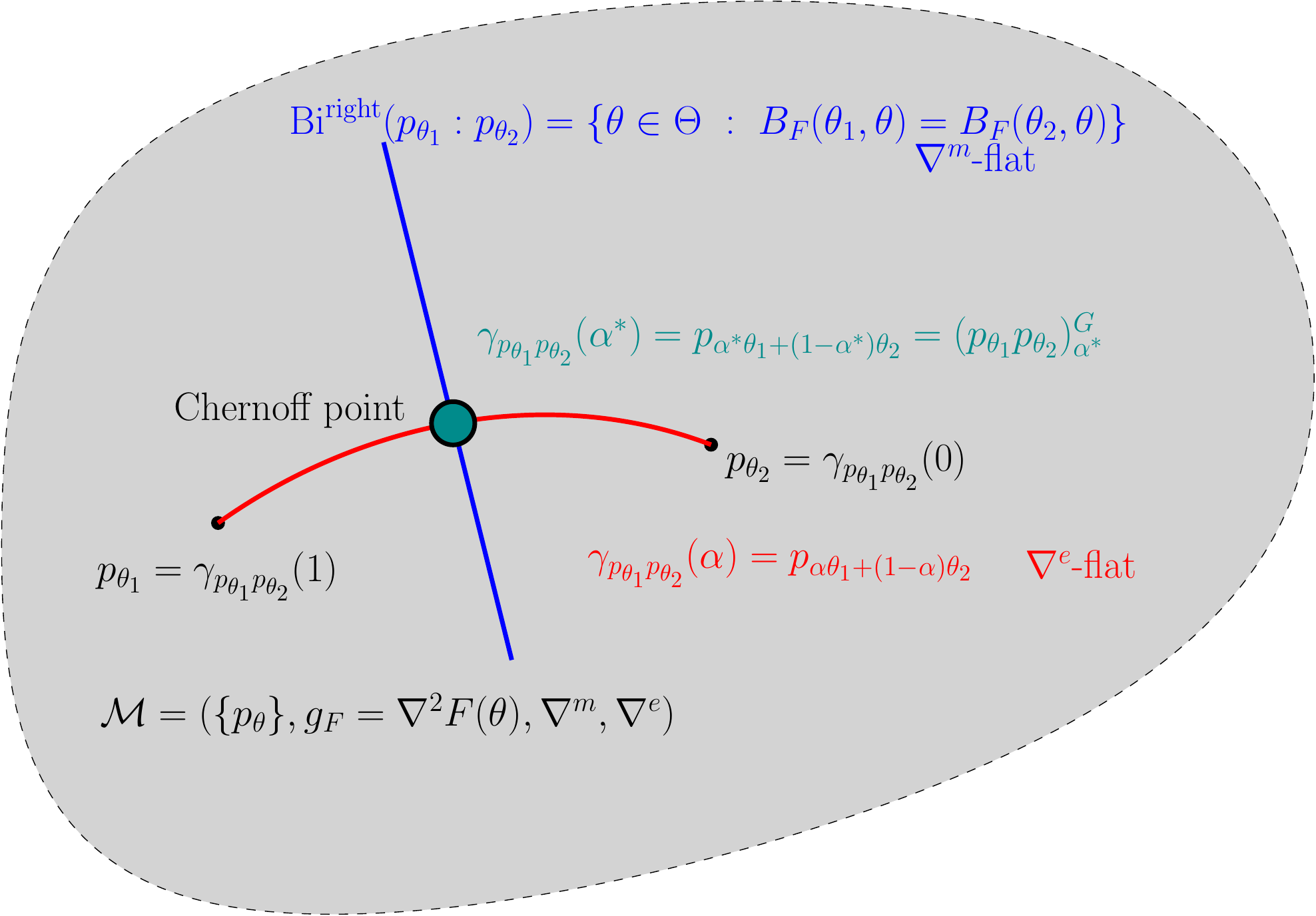}

\caption{The Chernoff information optimal skewing parameter $\alphastar$ for two densities $p_{\theta_1}$ and $p_{\theta_2}$ of some regular exponential family $\calE$ inducing an exponential family dually flat manifold $\mathcal{M}=(\{p_\theta\},g_F=\nabla^2 F(\theta),\nabla^m,\nabla^e)$ is characterized by the intersection of their $\nabla^e$-flat exponential geodesic  with their mixture bisector a $\nabla^m$-flat right-sided Bregman bisector.
\label{fig:ChernoffPoint}}
\end{figure}

\subsection{Case of one-dimensional parameters}
When the exponential family has one-dimensional natural parameter $\alpha\in\Theta\subset\bbR$, we thus get from $\OC_\SEF$:
$$
\eta_\alphastar=\frac{F(\alpha_2)-F(\alpha_1)}{\alpha_2-\alpha_1}.
$$
That is, $\alphastar$ can be obtained as the following closed-form formula:
\begin{equation}\label{eq:cf1d}
\alphastar=\frac{{F'}^{-1}\left(\frac{F(\alpha_2)-F(\alpha_1)}{\alpha_2-\alpha_1}\right)-\alpha_2}{\alpha_1-\alpha_2}.
\end{equation}

For multi-dimensional parameters $\theta$, we may consider the one-dimensional LREF $\calE_{p_{\theta_1}p_{\theta_2}}$ induced by $p_{\theta_1}$ and $p_{\theta_2}$ with $F_{\theta_1,\theta_2}(\alpha)=F((1-\alpha)\theta_1+\alpha \theta_2)$,
and write $F_{pq}'(\alpha)$ as the following directional derivative:
\begin{eqnarray}
\nabla_{\theta_2-\theta_1} F_{\theta_1,\theta_2}(\alpha) &\eqdef&
\lim_{\epsilon\rightarrow 0} 
\frac{F(\theta_1+(\epsilon+\alpha)(\theta_2-\theta_1))-F(\theta_1+\alpha(\theta_2-\theta_1))}{\epsilon},\\
&=& (\theta_2-\theta_1)^\top \nabla F(\theta_1+\alpha(\theta_2-\theta_1)),
\end{eqnarray}
using a first-order Taylor expansion.
Thus the optimality condition 
$$
\OC_{\SEF'} :\quad F_{\theta_1,\theta_2}'(\alpha)=0
$$ 
amounts to
\begin{equation}
\OC_\SEF:\quad (\theta_2-\theta_1)^\top \nabla F(\theta_1+\alphastar(\theta_2-\theta_1))=F(\theta_2)-F(\theta_1).
\end{equation}
This is equivalent to Eq. (8) of~\cite{CI-2013}.

\begin{Remark}
In general, we may consider multivariate Bregman divergences as univariate Bregman divergences:
We have 
\begin{equation}
B_F(\theta_1:\theta_2)=B_{F_{\theta_1,\theta_2}}(0:1),\quad\forall \theta_1,\theta_2\in\Theta
\end{equation}
where 
\begin{equation}
F_{\theta_1,\theta_2}(u) := F(\theta_1+u(\theta_2-\theta_1)).
\end{equation}
The functions $F_{\theta_1,\theta_2}$ are 1D Bregman generators (i.e., strictly convex and $C^1$), and we have the directional derivative
\begin{eqnarray*}
\nabla_{\theta_2-\theta_1} F_{\theta_1,\theta_2}(u) &=&
\lim_{\epsilon\rightarrow 0} 
\frac{F(\theta_1+(\epsilon+u)(\theta_2-\theta_1))-F(\theta_1+u(\theta_2-\theta_1))}{\epsilon},\\
&=& (\theta_2-\theta_1)^\top \nabla F(\theta_1+u(\theta_2-\theta_1)),
\end{eqnarray*}
Since $F_{\theta_1,\theta_2}(0)=F(\theta_1)$, $F_{\theta_1,\theta_2}(1)=F(\theta_2)$, and $F_{\theta_1,\theta_2}'(u)=\nabla_{\theta_2-\theta_1} F_{\theta_1,\theta_2}(u)$, it follows that
\begin{eqnarray*}
B_{F_{\theta_1,\theta_2}}(0:1)&=&F_{\theta_1,\theta_2}(0)-F_{\theta_1,\theta_2}(1)-(0-1)\nabla_{\theta_2-\theta_1} F_{\theta_1,\theta_2}(1),\\
&=& F(\theta_1)-F(\theta_2)+(\theta_2-\theta_1)^\top \nabla F(\theta_2) = B_F(\theta_1:\theta_2).
\end{eqnarray*}
Similarly, we can reparameterize Bregman divergences on a $k$-dimensional simplex by  $k$-dimensional Bregman divergences.
\end{Remark}

\begin{Remark}
Closing the loop: The Chernoff information although obtained from the     one-dimensional  likelihood ratio exponential family  yields as a corollary the general multi-parametric exponential families which as a special instance includes the one-dimensional exponential families (e.g, LREFs!).
\end{Remark}

\section{Forward and reverse Chernoff--Bregman divergences}\label{sec:CBD}

In this section, we shall define Chernoff-type symmetrizations of Bregman divergences inspired by the study of Chernoff information, and briefly mention  applications of these Chernoff--Bregman divergences in information theory.

\subsection{Chernoff--Bregman divergence}\label{sec:FCBD}

Let us define a Chernoff-like symmetrization of Bregman divergences~\cite{chen2008metrics2} different from the traditional Jeffreys--Bregman symmetrization:
\begin{eqnarray*}
B^J_F(\theta_1:\theta_2) &=& B_F(\theta_1:\theta_2)+B_F(\theta_2:\theta_1),\\
&=& (\theta_1-\theta_2)^\top (\nabla F(\theta_1)-\nabla F(\theta_2)),
\end{eqnarray*}
 or Jensen--Shannon-type symmetrization~\cite{nielsen2019jensen,nielsen2021variational} which yields a Jensen divergence~\cite{burbea1982convexity}:
\begin{eqnarray*}
B^\JS_F(\theta_1:\theta_2) &=& \frac{1}{2}\left(B_F\left(\theta_1:\frac{\theta_1+\theta_2}{2}\right)+B_F\left(\theta_2:\frac{\theta_1+\theta_2}{2}\right)\right),\\
&=& \frac{F(\theta_1)+F(\theta_2)}{2}-F\left(\frac{\theta_1+\theta_2}{2}\right)  =: J_F(\theta_1,\theta_2).
\end{eqnarray*}

\begin{Definition}[Chernoff--Bregman divergence]\label{def:CBD}
Let the Chernoff symmetrization of Bregman divergence $B_F(\theta_1;\theta_2)$ be the forward Chernoff--Bregman divergence $C_F(\theta_1,\theta_2)$ defined by 
\begin{equation}\label{eq:bcd}
C_F(\theta_1,\theta_2)=\max_{\alpha\in(0,1)} J_{F,\alpha}(\theta_1:\theta_2),
\end{equation}
where $J_{F,\alpha}$ is the $\alpha$-skewed Jensen divergence.
\end{Definition}

The optimization problem in Eq.~\ref{eq:bcd} may be equivalently rewritten~\cite{chen2008metrics2} as $\min_\theta R$ such that
both $B_F(\theta_1:\theta)\leq R$ and $B_F(\theta_2:\theta)\leq R$.
Thus the optimal value of $\alpha$ defines the circumcenter $\theta^*=\alpha\theta_1+(1-\alpha)\theta_2$ of 
the minimum enclosing right-sided Bregman sphere~\cite{nock2005fitting,nielsen2008smallest} and the Chernoff--Bregman  divergence:
$$
C_F(\theta_1,\theta_2)=\min_\theta \{B_F(\theta_1:\theta),B_F(\theta_2:\theta)\},
$$
corresponds to the radius of a minimum enclosing Bregman ball.
To summarize, this Chernoff symmetrization is a min-max symmetrization, and we have the following identities:
\begin{eqnarray*}
C_F(\theta_1,\theta_2) &=& \min_\theta \{B_F(\theta_1:\theta),B_F(\theta_2:\theta)\},\\
&=& \min_{\theta\in\Theta} \{\alpha B_F(\theta_1:\theta)+(1-\alpha) B_F(\theta_2:\theta)\},\\
&=& \max_{\alpha\in(0,1)} \{\alpha B_F(\theta_1:\alpha\theta_1+(1-\alpha)\theta_2)+(1-\alpha) B_F(\theta_2:\alpha\theta_1+(1-\alpha)\theta_2)\},\\
&=& \max_{\alpha\in(0,1)} J_{F,\alpha}(\theta_1:\theta_2).
\end{eqnarray*}
The second identity shows that the Chernoff symmetrization can be interpreted as a variational Jensen--Shannon-type divergence~\cite{nielsen2021variational}.

Notice that in general $C_F(\theta_1,\theta_2)\not=C_{F^*}(\eta_1,\eta_2)$ because the primal and dual geodesics do not coincide.
Those geodesics coincide only for symmetric Bregman divergences which are squared Mahalanobis divergences~\cite{BVD-2010}.

When $F(\theta)=F_\Shannon(\theta)=\sum_{i=1}^D \theta_i\log\theta_i$ (discrete Shannon negentropy), the Chernoff--Bregman  divergence is related to
the capacity of a discrete memoryless channel in information theory~\cite{chen2008metrics2,cover1999elements}.

Conditions for which $C_F(\theta_1,\theta_2)^a$ (with $a>0$) becomes a metric have been studied in~\cite{chen2008metrics2}: 
For example, $C_{F_\Shannon}^{\frac{1}{e}}$ is a metric distance~\cite{chen2008metrics2} (i.e., $a=\frac{1}{e}\simeq 0.36787944117$).
It is also known that the square root of the Chernoff distance between two univariate normal distributions is a metric distance~\cite{costa2016information}.

We can thus use the Bregman generalization of the Badoiu-Clarkson (BC) algorithm~\cite{nock2005fitting}  to compute an approximation of the smallest enclosing Bregman ball which in turn yields an approximation of the Chernoff--Bregman divergence:

Start from the initialization $\theta^{(0)}=\frac{\theta_1+\theta_2}{2}$, set $i\leftarrow 1$, and iterate as follows:
Let $f_i=\arg \max_{i\in\{1,2\}} B_F(\theta_i:\theta)$ and update the circumcenter by the following convex combination:
\begin{equation}\label{eq:walk}
\theta^{(i)}=\frac{i}{i+1}\theta^{(i-1)}+\frac{1}{i+1}\theta_{f_i}.
\end{equation}
This update corresponds to walking on the exponential arc $(\theta^{(i-1)}\theta_{f_i})^G$.
Notice that when there are only two points to compute their smallest enclosing Bregman ball, all the arcs $(\theta^{(i-1)}\theta_{f_i})^G$ are sub-arcs of the exponential arc $(\theta_1\theta_2)^G$.
See~\cite{nock2005fitting} for convergence results of this iterative algorithm.
Let us notice that Algorithm~1 approximates $\alphastar$ while the Bregman BC algorithm approximates in spirit $D_C(\theta_1,\theta_2)$ (and as a byproduct $\alphastar$).

\begin{Remark}
To compute the farthest point to the current circumcenter with respect to Bregman divergence, we need to find the sign of
$$
B_F(\theta_2:\theta)-B_F(\theta_1:\theta)=F(\theta_2)-F(\theta_1)-(\theta_2-\theta_1)\nabla F(\theta).
$$  
Thus we need to pre-calculate only once $F(\theta_1)$ and $F(\theta_2)$ which can be costly (e.g., $-\log \Mdet{\Sigma}$ functions need to be calculated only once when approximating the Chernoff information between Gaussians).
\end{Remark}

\subsection{Reverse Chernoff--Bregman divergence and universal coding}\label{sec:CIUC}
Similarly, we may define the reverse Chernoff--Bregman divergence by considering the minimum enclosing left-sided Bregman ball:
$$
C_F^R(\theta_1,\theta_2)=\min_\theta \{B_F(\theta:\theta_1),B_F(\theta:\theta_2)\}.
$$
Thus the reverse Bregman Chernoff divergence $D_C^R[\theta_1,\theta_2]=R^*$ is the radius of a minimum enclosing left-sided Bregman ball.

This reverse Chernoff--Bregman divergence finds application in universal coding in information theory (chapter 13 of~\cite{cover1999elements}, pp. 428-433):
Let $\calX=\{A_1, \ldots, A_d\}$ be a finite discrete alphabet of $d$ letters, and $X$ be a random variable with probability mass function 
$p$ on $\calX$.
Let $p_\lambda(x)$ denote the categorical distribution corresponding to $X$ 
so that $\Pr(X=A_i)=p_\lambda(A_i)$ with $\lambda=(\lambda^1,\ldots,\lambda^d)\in\bbR_{++}^d$ and $\sum_{i=1}^d \lambda^i=1$.
The Huffman codeword for $x\in\calX$ is of length $l(x)=-\log p(x)$ (ignoring integer ceil rounding), 
and the expected codeword length of $X$ is thus given by Shannon's entropy $H(X)=-\sum_x p(x)\log p(x)$.

If we code according to a distribution $p_{\lambda'}$ instead of the true distribution $p_\lambda$, the code is not optimal, and the
 redundancy $R(p_{\lambda},p_{\lambda'})$ is defined as 
the difference between the expected lengths of the codewords for $p_{\lambda'}$ and $p_{\lambda}$:
$$
R(p_{\lambda},p_{\lambda'})= (-E_{p_{\lambda}}[\log p_{\lambda'}(x)]  - (-E_{p_{\lambda}}[\log p_{\lambda}(x)])
= D_\KL[p_{\lambda}:p_{\lambda'}]\geq 0,
$$
where $D_\KL$ is the Kullback--Leibler divergence.

Now, suppose that the true distribution $p_{\lambda}$ belong to one of two prescribed distributions that we do not know: 
$p_\lambda \in \calP=\{p_{\lambda_1},  p_{\lambda_2} \}$.
Then we seek for the minimax redundancy:
$$
R^*= \min_{p_{\lambda}} \max_{i\in \{1,2\}} D_\KL[p_{\lambda_i}:p_{\lambda}].
$$
The distribution ${p_{\lambda^*}}$ achieving the minimax redundancy is the circumcenter of the right-centered KL ball enclosing the distributions $\calP$.
Using the natural coordinates $\theta=(\theta_1,\ldots,\theta_D)\in \bbR^D$ with $\theta_i=\log\frac{\lambda^i}{\lambda^d}$ of the  log-normalizer of the categorical  distributions (an exponential family of order $D=d-1$),
 we end up with calculating the smallest left-sided Bregman enclosing ball for the Bregman generator~\cite{EF-2009}: 
$F_{\mathrm{Categorical}}(\theta)=\log (1+\sum_{i=1}^D \exp\theta_i)$:
$$
R^* = \min_{\theta\in \in \bbR^D} \max_{i\in \{1,2\}}  B_{F_{\mathrm{categorical}}}(\theta:\theta_i).
$$
This latter minimax problem is unconstrained since $\theta\in\bbR^D=\bbR^{d-1}$.

\section{Chernoff information between Gaussian distributions}\label{sec:CIGauss}

\subsection{Invariance of Chernoff information under the action of the affine group}\label{sec:CIGaussInv}
The $d$-variate Gaussian density $p_\lambda(x)$ with parameter $\lambda=(\lambda_v=\mu,\lambda_M=\Sigma)$ where $\mu\in\bbR^d$ denotes the mean ($\mu=E_{p_\lambda}[x]$) and $\Sigma$ is a positive-definite covariance matrix ($\Sigma=\mathrm{Cov}_{p_\lambda}[X]$ for $X\sim p_\lambda$) is given by
$$
p_\lambda(x;\lambda) =  \frac{1}{(2\pi)^{\frac{d}{2}}\sqrt{|\lambda_M|}}  \exp\left(-\frac{1}{2} (x-\lambda_v)^\top \lambda_M^{-1} (x-\lambda_v)\right),
$$ 
where $|\cdot|$ denotes the matrix determinant.
The set of $d$-variate Gaussian distributions form a regular (and hence steep) exponential family with natural parameters
$\theta(\lambda)=\left(\lambda_M^{-1}\lambda_v,\frac{1}{2}\lambda_M^{-1}\right)$ and sufficient statistics $t(x)=(x,xx^\top)$.

The Bhattacharrya distance between two multivariate Gaussians distributions $p_{\mu_1,\Sigma_1}$ and $p_{\mu_2,\Sigma_2}$ is
$$
D_{B,\alpha}[p_{\mu_1,\Sigma_1},p_{\mu_2,\Sigma_2}] = 
\frac{1}{2}\left(\alpha\mu_1^\top\Sigma_1^{-1}\mu_1+(1-\alpha)\mu_2^\top\Sigma_2^{-1}\mu_2-\mu_\alpha^\top\Sigma_\alpha^{-1}\mu_\alpha
+\log \frac{|\Sigma_1|^{\alpha}|\Sigma_2|^{1-\alpha}}{|\Sigma_\alpha|}\right),
$$
where
\begin{eqnarray*}
\Sigma_\alpha &=& (\alpha\Sigma_1^{-1}+(1-\alpha)\Sigma_2^{-1})^{-1},\\
\mu_\alpha &=&  \Sigma_\alpha (\alpha\Sigma_1^{-1}\mu_1+(1-\alpha)\Sigma_2^{-1}).
\end{eqnarray*}

The Gaussian density can be rewritten as a multivariate location-scale family:
$$
p_\lambda(x;\lambda) = |\lambda_M|^{-\frac{1}{2}}\, p_\std(\lambda_M^\frac{1}{2}(x-\lambda_v)),
$$
where 
$$
p_\std(x)= \frac{1}{(2\pi)^{\frac{d}{2}}} \exp\left(-\frac{1}{2}x^\top x\right)=p_{(0,I)}
$$
denotes the standard multivariate Gaussian distribution.
The matrix $\lambda_M^\frac{1}{2}$ is the unique symmetric square-root matrix which is positive-definite when $\lambda_M$ is positive-definite.

\begin{Remark}
Notice that the product of two symmetric positive-definite matrices $P_1$ and $P_2$ may not be symmetric but
$P_1^{\frac{1}{2}}P_2P_1^{\frac{1}{2}}$ is always symmetric positive-definite, and the eigenvalues of $P_1^{\frac{1}{2}}P_2P_1^{\frac{1}{2}}$ coincides with the eigenvalues of $P_1P_2$.
Hence, we have
 $\lambda_\sp(P_1^{-\frac{1}{2}}P_2P_1^{-\frac{1}{2}})=\lambda_\sp(P_1^{-1}P_2)$ where $\lambda_\sp(M)$ denotes the eigenspectrum of matrix $M$.
\end{Remark}

We may interpret the Gaussian family  as obtained by the action $.$ of the affine group $\Aff(\bbR^d)=\bbR^d \rtimes \GL_d(\bbR)$ on the standard density $p_\std$: $p_{(\mu,\Sigma)}(x)=(\mu,\Sigma^{-\frac{1}{2}}) . p_\std(x)$
The affine group  is equipped with the following (outer) semidirect product:
$$
(l_1,A_1).(l_2,A_2)=(l_1+A_1l_2,A_1A_2),
$$
and this group can be handled as a matrix group with the following mapping of its elements to matrices:
$$
(l,A)\equiv \matrixtwotwo{A}{l}{0}{1}.
$$

We can show the following invariance of the skewed Bhattacharyya divergences:

\begin{Proposition}[Invariance of the Bhattacharyya divergence and $f$-divergences under the action of the affine group]\label{prop:invBhat}
We have 
\begin{eqnarray*}
D_{B,\alpha}[(\mu,\Sigma^{-\frac{1}{2}}) .p_{\mu_1,\Sigma_1}:(\mu,\Sigma^{-\frac{1}{2}}) . p_{\mu_2,\Sigma_2}] &=&
D_{B,\alpha}\left[p_{\Sigma^{-\frac{1}{2}}(\mu_1-\mu):\Sigma^{-\frac{1}{2}}\Sigma_1\Sigma^{-\frac{1}{2}}},p_{\Sigma^{-\frac{1}{2}}(\mu_2-\mu),\Sigma^{-\frac{1}{2}}\Sigma_2\Sigma^{-\frac{1}{2}}}\right],\\
&=&D_{B,\alpha}[p_{\mu_1,\Sigma_1}:p_{\mu_2,\Sigma_2}].
\end{eqnarray*}
\end{Proposition}

\begin{proof}
The proof follows from the $(f,g)$-form of Ali and Silvey's divergences~\cite{ali1966general}.
We can express $D_{B,\alpha}[p:q]=g(I_{h_\alpha}[p:q])$ where $h_\alpha(u)=-u^\alpha$ (convex for $\alpha\in(0,1)$) and $g(v)=-\log -v$.
Then we rely on the proof of invariance of $f$-divergences under the action of the affine group (see Proposition~3 of~\cite{nielsen2022note} relying on a change of variable in the integral):
$$
I_f[p_{\mu_1,\Sigma_1}:p_{\mu_2,\Sigma_2}]
= I_f\left[p_{0,I},p_{\Sigma_1^{-\frac{1}{2}}(\mu_2-\mu_1):\Sigma^{-\frac{1}{2}}_1\Sigma_2\Sigma^{-\frac{1}{2}}_1}\right]
= I_f\left[p_{\Sigma_2^{-\frac{1}{2}}(\mu_1-\mu_2),\Sigma_2^{-\frac{1}{2}}\Sigma_1\Sigma_2^{-\frac{1}{2}}}:p_{0,I}\right],
$$
where $I$ denotes the identity matrix.
\end{proof}

Thus by choosing $(\mu,\Sigma)=(\mu_1,\Sigma_1)$ and $(\mu,\Sigma)=(\mu_2,\Sigma_2)$, we obtain the following corollary:

\begin{Corollary}[Bhattacharyya divergence from canonical Bhattacharyya divergences]\label{prop:invCanBhat}
We have 
$$
D_{B,\alpha}[p_{\mu_1,\Sigma_1}:p_{\mu_2,\Sigma_2}]
= D_{B,\alpha}\left[p_{0,I},p_{\Sigma_1^{-\frac{1}{2}}(\mu_2-\mu_1):\Sigma_1^{-\frac{1}{2}}\Sigma_2\Sigma_1^{-\frac{1}{2}}}\right]
=D_{B,\alpha}\left[p_{\Sigma_2^{-\frac{1}{2}}(\mu_1-\mu_2):\Sigma_2^{-\frac{1}{2}}\Sigma_1\Sigma_2^{-\frac{1}{2}}},p_{0,I}\right].
$$
\end{Corollary}

It follows that the Chernoff optimal skewing parameter enjoys the same invariance property:
$$
\alphastar(p_{\mu_1,\Sigma_1}:p_{\mu_2,\Sigma_2})=\alphastar\left(p_{0,I},p_{\Sigma_1^{-\frac{1}{2}}(\mu_2-\mu_1),\Sigma_1^{-\frac{1}{2}}\Sigma_2\Sigma_1^{-\frac{1}{2}}}\right)
=\alphastar\left(p_{\Sigma_2^{-\frac{1}{2}}(\mu_1-\mu_2):\Sigma_2^{-\frac{1}{2}}\Sigma_1\Sigma_2^{-\frac{1}{2}}},p_{0,I}\right).
$$

As a byproduct, we get the invariance of the Chernoff information under the action of the affine group:

\begin{Corollary}[Invariance of the Chernoff information under the action of the affine group]\label{prop:invCher}
We have:
$$
D_C[p_{\mu_1,\Sigma_1},p_{\mu_2,\Sigma_2}]
= D_C\left[p_{0,I},p_{\Sigma_1^{-\frac{1}{2}}(\mu_2-\mu_1),\Sigma_1^{-\frac{1}{2}}\Sigma_2\Sigma_1^{-\frac{1}{2}}}\right]
=D_C\left[p_{\Sigma_2^{-\frac{1}{2}}(\mu_1-\mu_2),\Sigma_2^{-\frac{1}{2}}\Sigma_1\Sigma_2^{-\frac{1}{2}}},p_{0,I}\right].
$$
\end{Corollary}
			
Thus the formula for the Chernoff information between two Gaussians
$$
D_C(\mu_1,\Sigma_1,\mu_2,\Sigma_2):=D_C[p_{\mu_1,\Sigma_1},p_{\mu_2,\Sigma_2}]=D_C(\mu_{12},\Sigma_{12})
$$ 
can be written as a function of two terms 
$\mu_{12}=\Sigma_1^{-\frac{1}{2}}(\mu_2-\mu_1)$ and $\Sigma_{12}=\Sigma_1^{-\frac{1}{2}}\Sigma_2\Sigma_1^{-\frac{1}{2}}$.

\subsection{Closed-form formula for the Chernoff information between univariate Gaussian distributions}\label{sec:ChernoffUniGaussian}

We shall report the exact solution for the Chernoff information between univariate Gaussian distributions by solving a quadratic equation. We can also report a complex closed-form formula by using symbolic computing because the calculations are lengthy and thus prone to human error. 

\begin{sidewaysfigure}
\centering
\includegraphics[width=\textwidth]{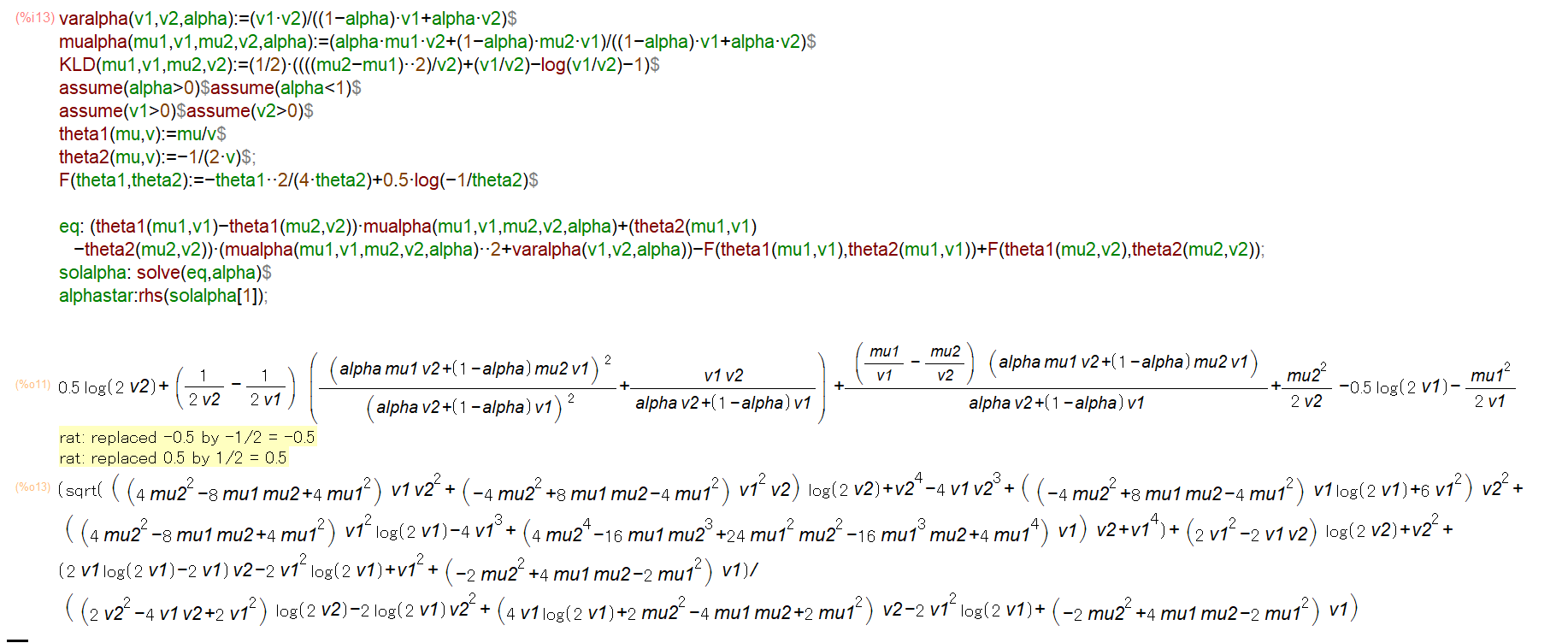}

\caption{The exact formula for the Chernoff skew value $\alphastar$ is reported in the cell \%o13 of this {\sc Maxima} session.\label{fig:exactsol}}
\end{sidewaysfigure}

Instantiating Eq.~\ref{eq:OCEF} for the case of univariate Gaussian distributions paramterized by $(\mu,\sigma^2)$, we get the following equation for the optimality condition of $\alphastar$:
\begin{eqnarray}
\inner{\theta_2-\theta_1}{\eta_\alphastar}&=&F(\theta_2)-F(\theta_1),\\
\Inner{\left(\frac{\mu_2}{\sigma_2^2}-\frac{\mu_1}{\sigma_1^2},\frac{1}{2\sigma_1^2}-\frac{1}{2\sigma_2^2}\right)}{(m_\alpha,v_\alpha)}&=&\frac{1}{2}\log\frac{\sigma_2^2}{\sigma_1^2}+\frac{\mu_2^2}{2\sigma_2^2}-\frac{\mu_1^2}{2\sigma_1^2},
\end{eqnarray}
where $\inner{\cdot}{\cdot}$ denotes the scalar product
and with the interpolated mean and variance along an exponential arc $\{(m_\alpha,v_\alpha^2)\}_{\alpha\in(0,1)}$ passing through $(\mu_1,\sigma_1^2)$ when $\alpha=1$ and  $(\mu_2,\sigma_2^2)$ when $\alpha=0$ given  by
\begin{eqnarray}
m_\alpha &=&\frac{\alpha\mu_1 \sigma_2^2+(1-\alpha)\mu_2 \sigma_1^2}{(1-\alpha)\sigma_1^2+\alpha \sigma_2^2}=
\frac{\alpha (\mu_1 \sigma_2^2-\mu_2 \sigma_1^2)+\mu_2 \sigma_1^2}{\sigma_1^2+\alpha (\sigma_2^2-\sigma_1^2)},\\
v_\alpha &=& \frac{\sigma_1^2\sigma_2^2}{(1-\alpha)\sigma_1^2+\alpha\sigma_2^2}=\frac{\sigma_1^2\sigma_2^2}{\sigma_1^2+\alpha (\sigma_2^2-\sigma_1^2)}.\label{eq:condug}
\end{eqnarray}
That is, for $p=p_{\mu_1,\sigma_1^2}$ and $q=p_{\mu_2,\sigma_2^2}$, we have the weighted geometric mixture $(pq)_\alpha^G=p_{m_\alpha,v_\alpha}$.

Thus the optimality condition of the Chernoff optimal skewing parameter is given by:
\begin{equation}\label{eq:OCGauss}
\OC_{\mathrm{Gaussian}}:\quad 
\left(\frac{\mu_2}{\sigma_2^2}-\frac{\mu_1}{\sigma_1^2}\right) m_\alpha
-\left(\frac{1}{2\sigma_2^2}-\frac{1}{2\sigma_1^2}\right) v_\alpha^2 =\frac{1}{2}\log\frac{\sigma_2^2}{\sigma_1^2}+\frac{\mu_2^2}{2\sigma_2^2}-\frac{\mu_1^2}{2\sigma_1^2}.
\end{equation}

Let us rewrite compactly Eq.~\ref{eq:OCGauss} as 
\begin{equation}\label{eq:Gaussian}
\OC_{\mathrm{Gaussian}}:\quad a_{12}m_\alpha + b_{12}v_\alpha + c_{12}=0,
\end{equation}
with the following coefficients:
\begin{eqnarray}
a_{12} &=& \frac{\mu_2}{\sigma_2^2}-\frac{\mu_1}{\sigma_1^2},\\
b_{12} &=& \frac{1}{2\sigma_1^2}-\frac{1}{2\sigma_2^2},\\
c_{12} &=& \frac{1}{2}\log\frac{\sigma_1^2}{\sigma_2^2}+\frac{\mu_1^2}{2\sigma_1^2}-\frac{\mu_2^2}{2\sigma_2^2}.
\end{eqnarray}

By multiplying both sides of Eq.~\ref{eq:Gaussian} by 
$\sigma_1^2+\alpha \Delta_{v}$ where $\Delta_{v}:=\sigma_2^2-\sigma_1^2$ and rearranging terms, we get a quadratic equation with positive root being $\alphastar$.

Using the computer algebra system (CAS) {\tt Maxima}, we can also solve  exactly this {\em quadratic equation} in $\alpha$ as 
a function of $\mu_1$ ,$\sigma_1^2$, $\mu_2$, and $\sigma_2^2$: 
See listing in Appendix~\ref{sec:maxima} and the screenshot of Figure~\ref{fig:exactsol}.

Once we get the optimal value of $\alphastar=\alphastar(\mu_1,\sigma_1^2,\mu_2,\sigma_2^2)$, we get the Chernoff information as
$$
D_C[p_{\mu_2,\sigma_2^2},p_{\mu_2,\sigma_2^2}]=D_\KL[p_{m_\alphastar,v_\alphastar}:p_{\mu_1,\sigma_1^2}]
$$
with
the Kullback-Leibler divergence between two univariate Gaussians distributions $p_{\mu_1,\sigma_1^2}$ and $p_{\mu_2,\sigma_2^2}$ given by
$$
D_\KL[p_{\mu_1,\sigma_1^2}:p_{\mu_2,\sigma_2^2}]= \frac{1}{2}\left(\frac{(\mu_2-\mu_1)^2}{\sigma_2^2}+\frac{\sigma_1^2}{\sigma_2^2}-\log \frac{\sigma_1^2}{\sigma_2^2}-1 \right).
$$

Notice that from the invariance of Proposition~\ref{prop:invBhat}, we have for any $(\mu,\sigma^2)\in\bbR\times\bbR_{++}$:
$$
D_\KL[p_{\mu_1,\sigma_1^2}:p_{\mu_2,\sigma_2^2}]=
D_\KL\left[
p_{\frac{\mu_1-\mu}{\sigma},\frac{\sigma_1^2}{\sigma^2}}
:
p_{\frac{\mu_2-\mu}{\sigma},\frac{\sigma_2^2}{\sigma^2}}
\right],
$$
and therefore by choosing $(\mu,\sigma^2)=(\mu_1,\sigma_1^2)$, we have
$$
D_\KL[p_{\mu_1,\sigma_1^2}:p_{\mu_2,\sigma_2^2}]= D_\KL\left[p_{0,1},p_{\frac{\mu_2-\mu_1}{\sigma_1},\frac{\sigma_2^2}{\sigma_1^2}}\right].
$$

\begin{Proposition}\label{prop:unigauss}
The Chernoff information between two univariate Gaussian distributions can be calculated exactly in closed form.
\end{Proposition}

Figure~\ref{fig:exactsol2} shows a snapshot of the obtained closed-form formula which is partially displayed in this window.
One can also program these closed-form solutions in Python using the {\tt SymPy} package (\url{https://www.sympy.org/en/index.html}) for performing symbolic computations.

\begin{sidewaysfigure}
\centering
\includegraphics[width=\textwidth]{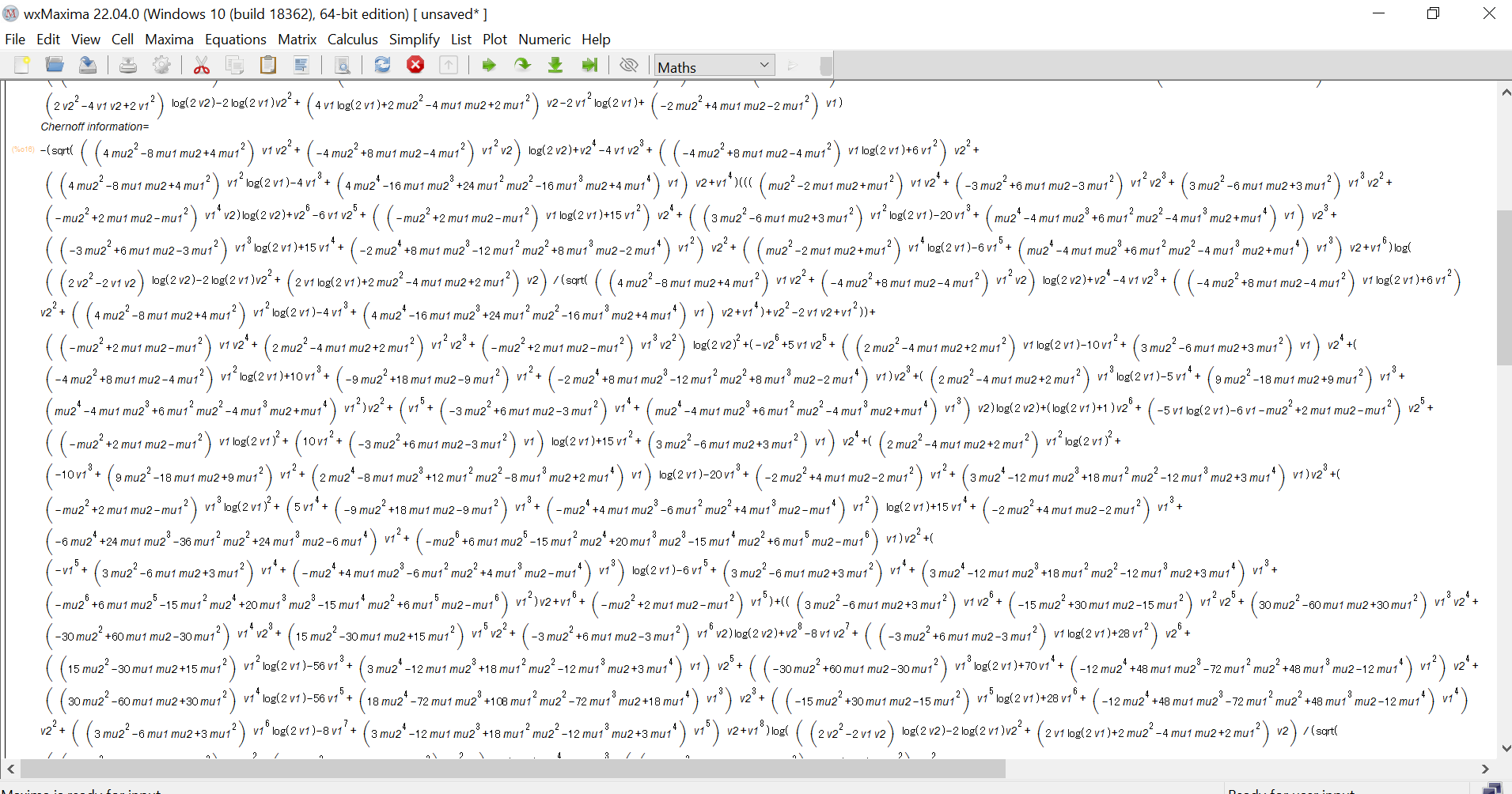}

\caption{The exact long exact formula for the Chernoff information is partially shown in this {\sc Maxima} session.\label{fig:exactsol2}}
\end{sidewaysfigure}

Let us report special cases with some illustrating examples.

\begin{itemize}
\item First, let us consider the Gaussian subfamily with prescribed variance.
When $\sigma_1^2=\sigma_2^2=\sigma^2$, we always have $\alphastar=\frac{1}{2}$, and the Chernoff information is
\begin{equation}
D_C[p_{\mu_1,\sigma^2}:p_{\mu_2,\sigma^2}]=\frac{(\mu_2-\mu_1)^2}{8\sigma^2}.
\end{equation} 
Notice that it amounts to one eight of the squared Mahalanobis distance (see~\cite{nielsen2022note} for a detailed explanation).

\item Second, let us consier Gaussian subfamily with prescribed mean.
When $\mu_1=\mu_2=\mu$, we get the optimal skewing value independent of the mean $\mu$:

{
$$
\alphastar=-{{v_{1}\,\log \left(2\,v_{2}\right)-v_{2}-v_{1}\,\log \left(2\,
 v_{1}\right)+v_{1}}\over{\left(v_{2}-v_{1}\right)\,\log \left(2\,
 v_{2}\right)-\log \left(2\,v_{1}\right)\,v_{2}+v_{1}\,\log \left(2\,
 v_{1}\right)}}
$$
}
where $v_1=\sigma_1^2$ and $v_2=\sigma_2^2$.
The Chernoff information is
{\small
\begin{equation}
D_C[p_{\mu_1,v_1}:p_{\mu_2,v_2}]-{{\left(v_{2}-v_{1}\right)\,\log \left({{v_{2}\,\log \left(2\,
 v_{2}\right)-\log \left(2\,v_{1}\right)\,v_{2}}\over{v_{2}-v_{1}}}
 \right)-v_{2}\,\log \left(2\,v_{2}\right)+\left(\log \left(2\,v_{1}
 \right)+1\right)\,v_{2}-v_{1}}\over{2\,v_{2}-2\,v_{1}}}.
\end{equation}
}

\item Third, consider the Chernoff information between the standard normal distribution and another normal distribution.
When $(\mu_1,\sigma_1^2)=(0,1)$ and $(\mu_2,\sigma_2^2)=(\mu,v)$, we get

\scalebox{0.8}{
$
\alphastar={{\sqrt{\left(4\,\mu^2\,v^2-4\,\mu^2\,v\right)\,\log \left(2\,v
 \right)+v^4-4\,v^3+\left(6-4\,\log 2\,\mu^2\right)\,v^2+\left(4\,\mu
 ^4+4\,\log 2\,\mu^2-4\right)\,v+1}+\left(2-2\,v\right)\,\log \left(2
 \,v\right)+v^2+\left(2\,\log 2-2\right)\,v-2\,\mu^2-2\,\log 2+1
 }\over{\left(2\,v^2-4\,v+2\right)\,\log \left(2\,v\right)-2\,\log 2
 \,v^2+\left(2\,\mu^2+4\,\log 2\right)\,v-2\,\mu^2-2\,\log 2}}
$
}
\end{itemize}

\begin{Example}\label{example1}
Let us consider $N(\mu_1=0,\sigma_1^2=1)$ and $N(\mu_2=1,\sigma_2^2=2)$.
The Chernoff exponent is
$$
\alphastar={{\sqrt{8\,\log 4-8\,\log 2+9}-2\,\log 4+2\,\log 2-1}\over{2\,\log
 4-2\,\log 2+2}}\approx 0.4215580558605244,
$$
and the Chernoff information is (zoom in for the formula):\vskip 0.3cm
\noindent\scalebox{0.5}{
$-{{\sqrt{8\,\log 4-8\,\log 2+9}\,\left(\left(2\,\log 4-2\,\log 2+3
 \right)\,\log \left({{4\,\log 4-4\,\log 2+4}\over{\sqrt{8\,\log 4-8
 \,\log 2+9}+1}}\right)-4\,\left(\log 4\right)^2+\left(8\,\log 2-6
 \right)\,\log 4-4\,\left(\log 2\right)^2+6\,\log 2-2\right)+\left(6
 \,\log 4-6\,\log 2+7\right)\,\log \left({{4\,\log 4-4\,\log 2+4
 }\over{\sqrt{8\,\log 4-8\,\log 2+9}+1}}\right)+4\,\left(\log 4
 \right)^2+\left(10-8\,\log 2\right)\,\log 4+4\,\left(\log 2\right)^2
 -10\,\log 2+6}\over{\left(4\,\log 4-4\,\log 2+6\right)\,\sqrt{8\,
 \log 4-8\,\log 2+9}+12\,\log 4-12\,\log 2+14}}$}

$$
\approx 0.1155433222682347
$$

\vskip 0.3cm

Using the bisection search of~\cite{CI-2013} with $\epsilon=10^{-8}$ takes $28$ iterations, and we get
$$\alphastar \approx 0.42155805602669716,
$$
and the Chernoff information is approximately $0.11554332226823472$.
Now, if we swap $p_{\mu_1,\sigma_1^2}\leftrightarrow p_{\mu_2,\sigma_2^2}$, we 
find $\alphastar=0.5784419439733028$ (and $0.5784419439733028+0.42155805602669716\approx 1$).
\end{Example}

Notice that in general, we may evaluate how good is the approximation $\tilde{\alpha}$ of $\alphastar$ by evaluating the deficiency of the 
optimal condition:
$$
\left|(\theta_2-\theta_1)^\top \eta_{\tilde\alpha}-F(\theta_2)+F(\theta_1)\right|.
$$

\begin{Example}
Let us consider $\mu_1=1$, $\sigma_1^2=3$ and $\mu_2=5$ and $\sigma_2^2=5$.
We get
$$
\alpha^*={{\sqrt{120\,\log 10-120\,\log 6+961}-3\,\log 10+3\,\log 6-23
 }\over{2\,\log 10-2\,\log 6+16}}\approx  0.4371453168322306
$$
and the Chernoff information is reported in closed form and evaluated numerically as
$$
0.5242883659200144.
$$
In comparison, the bisection algorithm of~\cite{CI-2013} with $\epsilon=10^{-8}$ takes $28$ iterations, 
and reports
$\alphastar\approx 0.43714531883597374$  and the Chernoff information about 
$$
0.5242883659200137.
$$
\end{Example}

\begin{Corollary}
The smallest enclosing left-sided Kullback-Leibler disk of $n$ univariate Gaussian distributions can be calculated exactly in randomized linear time~\cite{nielsen2008smallest}.
\end{Corollary}

\subsection{Fast approximation of the Chernoff information of multivariate Gaussian distributions}

In general, the Chernoff information between $d$-variate Gaussians distributions is not known in closed-form formula when $d>1$, see for example~\cite{athreya2017statistical,MVNChernoff-2018,tang2018limit}.
We shall consider below some special cases:

\begin{itemize}

\item When the Gaussians have the same covariance matrix $\Sigma$, the Chernoff information optimal skewing parameter is $\alpha=\frac{1}{2}$ and the Chernoff information is
$$
D_C[p_{\mu_1,\Sigma},p_{\mu_2,\Sigma}] = \frac{1}{8} \Delta^2_\Sigma(\mu_1,\mu_2),
$$
where $\Delta^2_\Sigma(\mu_1,\mu_2)=(\mu_2-\mu_1)^\top \Sigma^{-1} (\mu_2-\mu_1)$ is the squared Mahalanobis distance.
The Mahalanobis distance enjoys the following property by congruence transformation:
\begin{equation}
\Delta_\Sigma(\mu_1,\mu_2)=\Delta_{A\Sigma A^\top}(A\mu_1,A\mu_2), \forall A\in\GL(d).
\end{equation}
Notice that we can rewrite the (squared)  Mahalanobis distance as
$$
\Delta^2_\Sigma(\mu_1,\mu_2)=\tr\left(\Sigma^{-1} (\mu_2-\mu_1)(\mu_2-\mu_1)^\top\right)
$$
using the matrix trace cyclic property.
Then we check that 
\begin{eqnarray*}
\Delta_{A\Sigma A^\top}^2(A\mu_1,A\mu_2)&=&\tr\left(A^{-\top}\Sigma^{-1} A^{-1} A(\mu_2-\mu_1)(\mu_2-\mu_1)^\top A^\top\right),\\
&=& \tr(\Sigma^{-1} (\mu_2-\mu_1)(\mu_2-\mu_1)^\top )=\Delta_\Sigma^2(\mu_1,\mu_2).
\end{eqnarray*}

\item The Chernoff information for the special case of centered multivariate Gaussians distributions was studied in~\cite{MVNChernoff-2018}.
The KLD between two centered Gaussians $p_{\mu,\Sigma_1}$ and $p_{\mu_,\Sigma_2}$ is half of the matrix Burg distance:
\begin{equation}
D_\KL[p_{\mu,\Sigma_1}:p_{\mu,\Sigma_2}]=\frac{1}{2} \left(
\log\frac{\Mdet{\Sigma_2}}{\Mdet{\Sigma_1}}+\Mtr{\Sigma_2^{-1}\Sigma_1}-d
\right)=:\frac{1}{2}D_\Burg[\Sigma_1:\Sigma_2].
\end{equation}
When $d=1$, the Burg distance corresponds to the well-known Itakura-Saito divergence.
The matrix Burg distance is a matrix  spectral distance~\cite{MVNChernoff-2018}:
$$
D_\Burg[\Sigma_1:\Sigma_2] = \left( \sum_{i=1}^d \lambda_i-\log\lambda_i-1\right),
$$
where the $\lambda_i$'s are the eigenvalues of $\Sigma_2\Sigma_1^{-1}$.
The reverse KLD divergence $D_\KL[p_{\mu,\Sigma_2}:p_{\mu,\Sigma_1}]=\frac{1}{2}D_\Burg[\Sigma_2:\Sigma_1]$ 
is obtained by replacing $\lambda_i\leftrightarrow\frac{1}{\lambda_i}$:
$$
D_\KL[p_{\mu,\Sigma_2}:p_{\mu_,\Sigma_1}]=\frac{1}{2} \left( \sum_{i=1}^d \frac{1}{\lambda_i}+\log\lambda_i-1 \right).
$$
More generally, the $f$-divergences between centered Gaussian distributions are always matrix spectral divergences~\cite{nielsen2022note}.
\end{itemize}

Otherwise, for the general multivariate case, we implement the dichotomic search of Algorithm~1 in Algorithm~2 with the KLD between two multivariate Gaussian distributions expressed as
\begin{eqnarray}
D_\KL[p_{\mu_1,\Sigma_1}:p_{\mu_2,\Sigma_2}]&=&\frac{1}{2}\Delta^2_\Sigma(\mu_1,\mu_2)+\frac{1}{2}D_\Burg[\Sigma_1:\Sigma_2],\\
&=& \frac{1}{2}\left(\Mtr{\Sigma_2^{-1}\Sigma_1}-\log\frac{\Mdet{\Sigma_2}}{\Mdet{\Sigma_1}}-d+(\mu_2-\mu_1)^\top\Sigma_2^{-1}(\mu_2-\mu_1)\right).
\label{eq:kldmvn}
\end{eqnarray}

\noindent (Algorithm~2). Dichotomic search for approximating the Chernoff information between two multivariate normal distributions $p_{\mu_1,\Sigma_1}$ and $p_{\mu_2,\Sigma_2}$ by approximating the optimal skewing parameter value $\alpha\approx\alphastar$.\vskip 0.3cm
 \begin{algorithm} 
                \SetAlgoLined
                \SetKwInOut{Input}{input}
                \SetKwInOut{Ret}{return}
                \Input{Two normal densities $p_{\mu_1,\Sigma_1}$ and $p_{\mu_2,\Sigma_2}$, and a numerical precision threshold $\epsilon>0$}
								$\alpha_m=0$\; $\alpha_M=1$\;
                    \While{$|\alpha_M-\alpha_m|>\epsilon$}{
                    $\alpha=\frac{\alpha_m+\alpha_M}{2}$\;
											$\Sigma_\alpha^e=\left((1-\alpha) \Sigma_1^{-1}+\alpha\Sigma_2^{-1}\right)^{-1}$\;
											$\mu_\alpha^e=\Sigma_\alpha^e \left((1-\alpha)\Sigma_1^{-1}\mu_1+\alpha\Sigma_2^{-1}\mu_2\right)$\;
											\tcp{Formula of the KLD between two normal distributions in Eq.~\ref{eq:kldmvn}}
            \If{$D_\KL[p_{\mu_\alpha^e,\Sigma_\alpha^e}:p_{\mu_1,\Sigma_1}]>D_\KL[p_{\mu_\alpha^e,\Sigma_\alpha^e}:p_{\mu_2,\Sigma_2}$}{$\alpha_m=\alpha$\; \tcp{See Figure~\ref{fig:Dicho} for an illustration and Proposition~\ref{prop:GI}}}
						\Else{$\alpha_M=\alpha$\;}
															} 
                    \Return{ 
                     $D_\KL[p_{\mu_\alpha^e,\Sigma_\alpha^e}:p_{\mu_1,\Sigma_1}]$\;
                    }
            \end{algorithm}

\begin{Example}
Let $d=2$, $p_{\mu_1,\Sigma_1}=p_{0,I}$ be the standard bivariate Gaussian distribution and 
$p_{\mu_2,\Sigma_2}$ be the bivariate Gaussian distribution with mean $\mu_2=[1\ 2]^\top$ and covariance matrix $\Sigma_2=\matrixtwotwo{1}{-1}{-1}{2}$.
Setting the numerical precision threshold $\epsilon$ to $\epsilon=10^{-8}$, the  dichotomic search performs $28$ split iterations, and approximate $\alphastar$ by
$$
\alphastar\approx 0.5825489424169064.
$$
The  Chernoff information $D_C[p_{0,I},p_{\mu_2,\Sigma_2}]$ is approximated by $0.8827640697808525$.
\end{Example}

The $m$-interpolation of multivariate Gaussian distributions $p_{\mu_1,\Sigma_1}$ and $p_{\mu_2,\Sigma_2}$ with respect to the mixture connection $\nabla^m$ is given by
$$
\gamma^m_{p_{\mu_1,\Sigma_1},p_{\mu_2,\Sigma_2}}(\alpha)=p_{\mu_\alpha^m,\Sigma_\alpha^m},
$$ 
where
\begin{eqnarray*}
\mu_\alpha^m&=&(1-\alpha)\mu_1+\alpha\mu_2=:\bar\mu_\alpha,\\
\Sigma_\alpha^m &=& (1-\alpha)\Sigma_1+\alpha\Sigma_2+(1-\alpha)\mu_1\mu_1^\top +\alpha\mu_2\mu_2^\top-\bar\mu_\alpha\bar\mu_\alpha^\top.
\end{eqnarray*}

To $e$-interpolation of multivariate Gaussian distributions $p_{\mu_1,\Sigma_1}$ and $p_{\mu_2,\Sigma_2}$ with respect to the exponential connection $\nabla^e$ is given by
$$
\gamma^e_{p_{\mu_1,\Sigma_1},p_{\mu_2,\Sigma_2}}(\alpha)=p_{\mu_\alpha^e,\Sigma_\alpha^e},
$$
where
\begin{eqnarray*}
\mu_\alpha^e&=& \Sigma_\alpha^e \left((1-\alpha)\Sigma_1^{-1}\mu_1+\alpha\Sigma_2^{-1}\mu_2\right),\\
\Sigma_\alpha^e &=&  \left((1-\alpha) \Sigma_1^{-1}+\alpha\Sigma_2^{-1}\right)^{-1}.
\end{eqnarray*}
In information geometry, both these $e$- and $m$-connections defined with respect to an exponential family are shown to be flat.
These geodesics correspond to linear interpolations in the $\nabla^e$-affine coordinate system $\theta$ and in the dual $\nabla^m$ coordinate system $\eta$, respectively.

Figure~\ref{fig:MVNinterpolation} displays these two $e$-geodesic and $m$-geodesic between two multivariate normal distributions.
Notice that the Riemannian geodesic with the Levi-Civita metric connection $\frac{\nabla^e+\nabla^m}{2}$ is not known in closed form for boundary value conditions. The expression of the Riemannian geodesic is known only for initial value conditions~\cite{calvo1991explicit} (i.e., starting point with a given vector direction).

\begin{figure}
\centering
\includegraphics[width=0.75\textwidth]{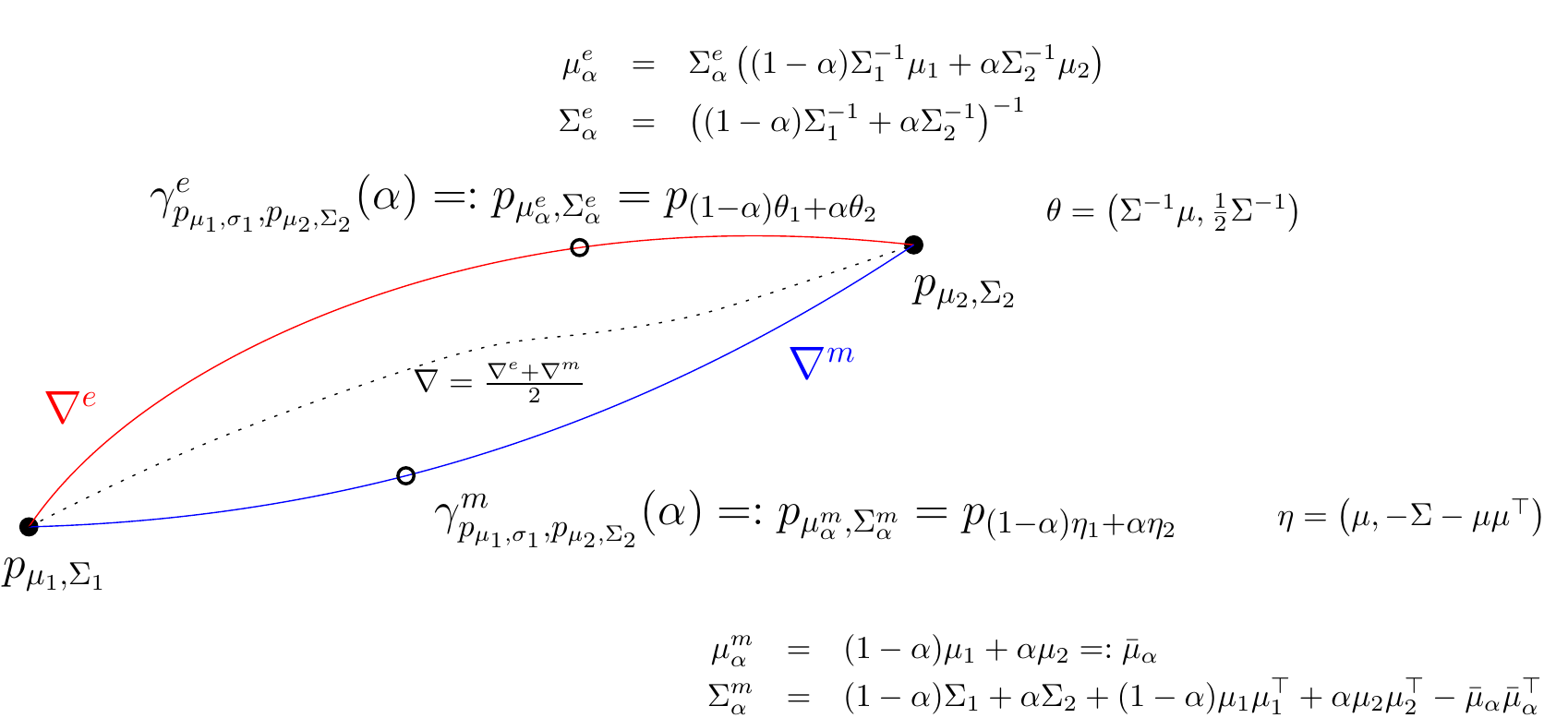}
\caption{Interpolation along the $e$-geodesic and the $m$-geodesic passing through two given multivariate normal distributions.
No closed-form is known for Riemannian geodesic with respect to the metric Levi-Civita connection (shown in dashed style).}
\label{fig:MVNinterpolation}
\end{figure}

\subsection{Chernoff information between centered multivariate normal distributions}\label{sec:zeroMVN}

The set 
$$
\mathcal{N}_0=\left\{p_{\Sigma}(x)=\frac{1}{\sqrt{\Mdet{2\pi\Sigma}}} \exp\left(-\frac{1}{2}x^\top\Sigma^{-1}x\right) \st \Sigma\succ 0\right\}
$$ 
of centered multivariate normal distributions is a regular exponential family with natural parameter $\theta=\Sigma^{-1}$,
sufficient statistic $t(x)=-\frac{1}{2}xx^\top$, log-normalizer $F(\theta)=-\frac{1}{2}\log \Mdet{\theta}$ and auxiliary carrier term $k(x)=-\frac{d}{2}\log(2\pi)$. Family $\mathcal{N}_0$ is also a multivariate scale family with scale matrices $\Sigma^{\frac{1}{2}}$ (standard deviation $\sigma$ in 1D).

Let $\inner{A}{B}=\Mtr{A^\top B}$ defines the inner product between two symmetric matrices $A$ and $B$.
Then we can write the centered Gaussian distribution $p_{\Sigma}(x)$ in the canonical form of exponential families:
$$
p_{\theta}(x)=\exp\left(\inner{\theta}{t(x)}-F(\theta)+k(x)\right).
$$
The function log det of a positive-definite matrix is strictly concave~\cite{boyd2004convex}, and hence we check that $F(\theta)$ is strictly convex. 
Furthermore, we have $\nabla_X \log\Mdet{X}=X^{-\top}$ so that 
$\nabla_\theta F(\theta)=-\frac{1}{2}\theta^{-\top}$.

The optimality condition equation of Chernoff best skewing parameter $\alphastar$ becomes:

\begin{eqnarray}
\inner{\theta_2-\theta_1}{\nabla F(\theta_1+\alphastar(\theta_2-\theta_1))} &=& F(\theta_2)-F(\theta_1),\\
-\frac{1}{2}\Mtr{(\theta_2-\theta_1)^\top (\theta_1+\alphastar(\theta_2-\theta_1))^{-1}}&=&-\frac{1}{2}\log\frac{\Mdet{\theta_2}}{\Mdet{\theta_1}},\\
\Mtr{(\theta_2-\theta_1)^\top (\theta_1+\alphastar(\theta_2-\theta_1))^{-1}}&=&\log\frac{\Mdet{\theta_2}}{\Mdet{\theta_1}},\\
\Mtr{(\Sigma_2^{-1}-\Sigma_1^{-1})\ (\Sigma_1^{-1}+\alphastar(\Sigma_2^{-1}-\Sigma_1^{-1}))^{-1}}
&=&\log\frac{\Mdet{\Sigma_1}}{\Mdet{\Sigma_2}}=\log\Mdet{\Sigma_1\Sigma_2^{-1}}.\label{eq:cmvncond}
\end{eqnarray}

When $\Sigma_2=s\Sigma_1$ (and $\Sigma_2^{-1}=\frac{1}{s}\Sigma_1^{-1}$) for $s>0$ and $s\not=1$, we get a closed-form for $\alphastar$ 
using the fact that $\Mdet{\frac{I}{s}}=\frac{1}{s^d}$ and $\Mtr{I}=d$ for $d$-dimensional identity matrix $I$.
Solving Eq.~\ref{eq:cmvncond} yields
\begin{equation}
\alphastar(s)=\frac{s-1-\log s}{(s-1)\log s}\in(0,1).
\end{equation}
Therefore  the Chernoff information between two scaled centered Gaussian distributions $p_{\mu,\Sigma}$ and $p_{\mu,s\Sigma}$ is available in closed form.

\begin{Proposition}\label{prop:CIscaleCov}
The Chernoff information between two scaled $d$-dimensional centered Gaussian distributions $p_{\mu,\Sigma}$ and $p_{\mu,s\Sigma}$ of $\mathcal{N}_\mu$ (for $s>0$) is available in closed form: 
\begin{equation}\label{eq:CIscale}
D_C[p_{\mu,\Sigma},p_{\mu,s\Sigma}]= 
D_{B,\alphastar}[p_{\mu,\Sigma},p_{\mu,s\Sigma}]=
d\frac{(s-1)\log\left(\frac{s}{s-1}\log s\right)-s\log s+s-1}{2(1-s)},
\end{equation}
where $\alphastar=\frac{s-1-\log s}{(s-1)\log s}\in(0,1)$.
\end{Proposition}

Notice that $\alphastar(p_{\mu,\Sigma}:p_{\mu,s\Sigma})=\alphastar(p_{\mu,\Sigma},p_{\mu,\frac{1}{s}\Sigma})$ and
$D_C[p_{\mu,\Sigma},p_{\mu,s\Sigma}]=D_C[p_{\mu,\Sigma},p_{\mu,\frac{1}{s}\Sigma}]$.

\begin{Example}
Consider $\mu_1=\mu_2=0$ and $\Sigma_1=I$, $\Sigma_2=\frac{1}{2}I$.
We find that $\alphastar=\frac{2\log 2-1}{\log 2}$, which is independent of the dimension of the matrices.
The Chernoff information depends on the dimension:
$$
D_C[p_{0,I},p_{0,\frac{1}{2}I}]=d\frac{\log 2-\log\log 2-1}{2}.
$$
\end{Example}

Notice that when $d=1$, we have $s=\frac{\sigma_2^2}{\sigma_1^2}$, and we recover a special case of the closed-form formula for the Chernoff information between univariate Gaussians.

In~\cite{MVNChernoff-2018}, the following equation is reported for finding $\alphastar$ based on Eq.~\ref{eq:cmvncond}:
\begin{equation}\label{eq:CIzeroMVN}
\OC_{\mathrm{Centered Gaussians}}:\quad \sum_{i=1}^d \frac{1-\lambda_i}{\alphastar+(1-\alphastar)\lambda_i}+\log\lambda_i=0
\end{equation}
where the $\lambda_i$'s are generalized eigenvalues of $\Sigma_1\Sigma_2^{-1}$ (this excludes the case of all $\lambda_i$'s equal to one). 
The value of $\alphastar$ satisfying Eq.~\ref{eq:CIzeroMVN} is unique.
Let us notice that the product of two symmetric positive-definite matrices is not necessarily symmetric anymore.
We can derive Eq.~\ref{eq:CIzeroMVN} by expressing Eq.~\ref{eq:cmvncond} using the identity matrix $I$ and matrix 
$\Sigma_2^{-\frac{1}{2}}\Sigma_1\Sigma_2^{-\frac{1}{2}}$.

\begin{Remark}
We can get closed-form solutions for $\alphastar$ and the corresponding Chernoff information in some particular cases.
For example, when the dimension $d=2$, we need to solve a quadratic equation to get $\alphastar$. 
Thus for $d\leq 4$, we get a closed-form solution for $\alphastar$ by solving a polynomial equation characterizing the optimal condition, and obtain the Chernoff information in closed-form as a byproduct.
\end{Remark}

\begin{Example}
Consider the Chernoff information between $p_{0,I}$ and $p_{0,\Lambda}$ with $\Lambda=\diag(1,2,3,4)$.
We get the exact Chernoff exponent value $\alphastar$ by taking the root of a quartic polynomial equation falling in $(0,1)$.
By evaluating numerically this root, we find that $\alphastar\simeq 0.59694$ and that the Chernoff information is
$D_C[p_{0,I},p_{0,\Lambda}]\simeq 0.22076$. See Appendix for some symbolic computation code.
\end{Example}

\section{Chernoff information between densities of different exponential families}\label{sec:diffEF}

Let 
$$\calE_1=\{p_\theta=\exp(\inner{\theta}{t_1(x)}-F_1(\theta)) \st \theta\in\Theta\},
$$ and
$$
\calE_2=\{q_{\theta'}=\exp(\inner{\theta'}{t_2(x)}-F_2(\theta')\st \theta'\in\Theta'\},
$$ 
be two distinct exponential families, and consider the Chernoff information between the densities $p_{\theta_1}$ and $q_{\theta_2'}$.
The exponential arc induced by $p_{\theta_1}$ and $q_{\theta_2'}$ is
$$
\{ (p_{\theta_1}q_{\theta_2'})_\alpha^G \propto p_{\theta_1}^\alpha  q_{\theta_2'}^{1-\alpha} \st \alpha\in (0,1)\}.
$$
Let $\calE_{12}$ denote the exponential family with sufficient statistics $(t_1(x),t_2(x))$, log-normalizer $F_{12}(\theta,\theta')$, and denote by $\Theta_{12}$ its natural parameter space. Family $\calE_{12}$ can be interpreted as a product exponential family which yields an exponential family.
We have
$$
(p_{\theta_1}q_{\theta_2'})_\alpha^G=\exp\left(\inner{(t_1(x),t_2(x))}{(\alpha\theta_1,(1-\alpha)\theta_2')}-F_{12}(\alpha\theta_1,(1-\alpha)\theta_2')\right).
$$
Thus the induced LREF $\calE_{p_{\theta_1}q_{\theta_2'}}$ with natural parameter space $\Theta_{p_{\theta_1}q_{\theta_2'}}$ can be interpreted as a {\em 1D curved  exponential family of the product exponential family} $\calE_{12}$.

The optimal skewing parameter $\alphastar$ is found by setting the derivative of $F_{12}(\alpha\theta_1,(1-\alpha)\theta_2')$ with respect $\alpha$ to zero:
$$
\frac{d}{\mathrm{d}\alpha}F_{12}(\alpha\theta_1,(1-\alpha)\theta_2')=0.
$$

\begin{Example}
Let $\calE_1$ can be chosen as the exponential family of exponential distributions
$$
\calE_1=\left\{e_\lambda(x)=\lambda\exp(-\lambda x), \lambda\in(0,+\infty)\right\}
$$ 
defined on the support $\calX_1=(0,\infty)$ and
$\calE_2$ can be chosen as the exponential family of half-normal distributions
$$
\calE_2=\left\{h_\sigma(x)=\sqrt{\frac{2}{\pi\sigma^2}}\exp(-\frac{x^2}{2\sigma^2}) \st \sigma^2>0\right\}
$$ 
with support $\calX_2=(0,\infty)$.

The product exponential family corresponds to the singly truncated normal family~\cite{del1994singly} which is a non-regular (i.e., parameter space is not topologically an open set):
$$
\Theta_{12}=(\bbR\times\bbR_{++}) \cup\Theta_0,
$$
with $\Theta_0=\{(\theta,0) \st \theta<0\}$ (the part corresponding to the exponential family of exponential distributions).
This exponential family $\calE_{12}=\{p_{\theta_1,\theta_2}\}$ of singly truncated normal distributions  is also non-steep~\cite{del1994singly}.
The log-normalizer is
$$
F_{12}(\theta_1,\theta_2)=\frac{1}{2}\log\frac{\pi}{\theta_2}+\log \Phi\left(\frac{\theta_1}{\sqrt{2\theta_2}}\right)+\frac{\theta_1^2}{4\theta_2},
$$
where $\theta_1=\frac{\mu}{\sigma^2}$ and $\theta_2=\frac{1}{2\sigma^2}$, and $\Phi$ denotes the cumulative distribution function of the standard normal.
Function $F_{12}$ is proven of class $C^1$ on $\Theta_{12}$ (see Proposition 3.1 of~\cite{del1994singly}) with $F_{12}(\theta,0)=-\log(-\theta)$ for $\theta<0$.

Notice that the KLD between an exponential distribution and a half-normal distribution is $+\infty$ since the definite integral diverges (hence $D_\KL[e_\lambda:h_\sigma]$ is not equivalent to a Bregman divergence, and $\Theta_{e_{\theta_1}h_{\theta_2'}}$ is not open at $1$) but the reverse KLD between a  half-normal distribution and an exponential distribution is available in closed-form (using symbolic computing):
$$
D_\KL[h_\sigma:e_\lambda]=\sqrt{\sqrt{8}\sigma\lambda-\sqrt{\pi}(1+\log\frac{\pi\lambda^2\sigma^2}{2})}{2\sqrt{\pi}}.
$$

Figure~\ref{fig:halfdomain} illustrate the domain of the singly truncated normal distributions and displays an exponential arc between an exponential distribution and a half-normal distribution.
Notice that we could have also considered a similar but different example by taking the exponential family of Rayleigh distributions which exhibit an additional extra carrier term $k(x)$.

The Bhattacharyya $\alpha$-skewed coefficient calculated using symbolic computing (see Appendix) is

{\footnotesize
$$
\rho_\alpha[h_\sigma:e_\lambda]= \rho_{1-\alpha}[e_\lambda:h_\sigma]=
{{\pi^{{{1}\over{2}}-{{\alpha}\over{2}}}\,e^ {- \sigma^2\,\lambda^2
  }\,\left(2^{{{\alpha}\over{2}}+{{1}\over{2}}}\,\sigma\,\lambda\,e^{
 {{\alpha\,\sigma^2\,\lambda^2}\over{2}}+{{\sigma^2\,\lambda^2}\over{
 2\,\alpha}}}\,\mathrm{erf}\left({{\left(\sqrt{2}\,\alpha-\sqrt{2}
 \right)\,\sigma\,\lambda}\over{2\,\sqrt{\alpha}}}\right)+2^{{{\alpha
 }\over{2}}+{{1}\over{2}}}\,\sigma\,\lambda\,e^{{{\alpha\,\sigma^2\,
 \lambda^2}\over{2}}+{{\sigma^2\,\lambda^2}\over{2\,\alpha}}}\right)
 }\over{2\,\sqrt{\alpha}\,\sigma^{\alpha}\,\lambda^{\alpha}}}, 
$$
}
where $\erf$ denotes the error function.
\end{Example}

 \begin{figure}
\centering
\includegraphics[width=0.55\textwidth]{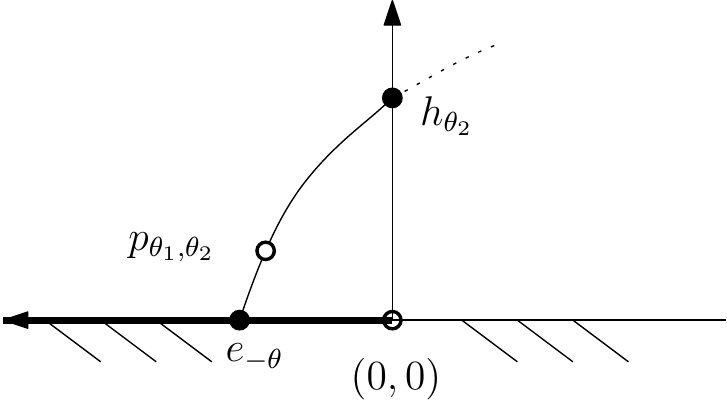}
\caption{The natural parameter space of the non-regular full exponential family of singly truncated normal distributions is not regular (i.e., not open): The negative real axis corresponds to the exponential family of exponential distributions.}
\label{fig:halfdomain}
\end{figure}

\section{Conclusion}\label{sec:concl}
In this work, we revisited the Chernoff information~\cite{chernoff1952measure} (1952) which was originally introduced to upper bound Bayes' error in binary hypothesis testing. A general characterization of Chernoff information between two arbitrary probability measures was given in~\cite{RenyiDiv-2014} (Theorem~32) by considering R\'enyi divergences which can be interpreted as scaled skewed Bhattacharyya divergences. 
Since its inception, the Chernoff information has proven useful as a statistical divergence (Chernoff divergence) in many applications ranging from information fusion to quantum metrology due to its empirical robustness property~\cite{julier2006empirical}.
Informally, we may observe  empirically that in practice the skewed Bhattacharyya divergence is more stable around the Chernoff exponent $\alphastar$ than in other part of the range $(0,1)$.
By considering the maximal extension of the exponential arc joining two densities $p$ and $q$ on a Lebesgue space $L^1(\mu)$, we built full likelihood ratio exponential families~\cite{ITEF-2007} $\calE_{pq}$ (LREFs) in \S\ref{sec:CILREF}. 
When the LREF $\calE_{pq}$ is a regular exponential family (with coinciding support of $p$ and $q$), both the forward and reverse Kullback--Leibler divergence are finite and can be rewritten as finite Bregman divergences induced by the log-normalizer $F_{pq}$ of $\calE_{pq}$ which amounts to minus skewed Bhattacharyya divergences.
Since  log-normalizers of exponential families are strictly convex, we deduced that the skewed Bhattacharyya divergences are strictly concave and their maximization yielding the Chernoff information is hence proven unique. 
As a byproduct, this geometric characterization in $L^1(\mu)$ allowed us to prove that the intersection of a $e$-geodesic with a $m$-bisector is unique in  dually flat subspaces of $L^1(\mu)$, 
and similarly that the intersection of a $m$-geodesic with a $e$-bisector is unique (Proposition~\ref{prop:iguniqueinter}).
We then considered the exponential families of univariate and multivariate normal distributions: 
We reported closed-form solutions for the Chernoff information of univariate normal distribution and centered normal distributions with scaled covariance matrices, and show how to implement efficiently a dichotomic search for approximating the Chernoff information between two multivariate normal distributions (Algorithm~2).
Table~\ref{tab:oc} summarizes the various optimal condition studied characterizing the Chernoff exponent.
Finally, inspired by the Chernoff information study, we defined in \S\ref{sec:CBD}, the forward and 
reverse Bregman--Chernoff divergences~\cite{chen2008metrics}, and show how these divergences are related to the capacity of a discrete memoryless channel
and the minimax redundancy of universal coding in information theory~\cite{cover1999elements}.

\begin{table}
\begin{tabular}{|ll|}
 \hline
\multicolumn{2}{|c|}{Generic case}\\ \hline
Primal LREF & $\OC_\alpha: D_\KL[(pq)^G_\alphastar:p]=D_\KL[(pq)^G_\alphastar:q]$ \\
Dual LREF & $\OC_\beta: \beta(\alphastar)=E_{(pq)_\alphastar^G}\left[\log\frac{p(x)}{q(x)}\right]=0$\\
Geometric OC & $(pq)^G_\alphastar=\gamma^G(p,q) \cap\Bi_\KL^\mathrm{left}(p,q]$  \\ \hline
\multicolumn{2}{|c|}{Case of exponential families}\\ \hline
Bregman  & $\OC_\EF:\quad B_F(\theta_1:\theta_\alphastar)=B_F(\theta_2:\theta_\alphastar)$\\
Fenchel-Young & $\OC_\YF:\quad Y_{F,F^*}(\theta_1:\eta_\alphastar)=Y_{F,F^*}(\theta_2:\eta_\alphastar)$\\
Simplified  & $\OC_{\SEF'} :\quad F_{\theta_1,\theta_2}'(\alpha)=0$ \\
  & $\OC_\SEF:\quad (\theta_2-\theta_1)^\top \nabla F(\theta_1+\alphastar(\theta_2-\theta_1))=F(\theta_2)-F(\theta_1)$\\ 
	Geometric OC &  $\gamma^e_{pq}(\alpha)\cap \Bi^m(p,q)$\\ \hline
\multicolumn{2}{|c|}{Gaussian case}\\ \hline
Univariate Gaussians & $\OC_{\mathrm{Gaussian}}:\quad 
\left(\frac{\mu_2}{\sigma_2^2}-\frac{\mu_1}{\sigma_1^2}\right) m_\alpha
-\left(\frac{1}{2\sigma_2^2}-\frac{1}{2\sigma_1^2}\right) v_\alpha^2 =\frac{1}{2}\log\frac{\sigma_2^2}{\sigma_1^2}+\frac{\mu_2^2}{2\sigma_2^2}-\frac{\mu_1^2}{2\sigma_1^2}$\\
Centered Gaussians  & 
 $\OC_{\mathrm{Centered Gaussians}}:\sum_{i=1}^d \frac{1-\lambda_i}{\alphastar+(1-\alphastar)\lambda_i}+\log\lambda_i=0$ \\
\hline
\end{tabular}

\caption{Summary of the optimal conditions characterizing the Chernoff exponent.}\label{tab:oc}
\end{table}

\vskip 0.33cm
\noindent Additional material including {\sc Maxima} and {\sc Java}\textregistered{} snippet codes is available online at\\
 \url{https://franknielsen.github.io/ChernoffInformation/index.html} 
 \vskip 0.33cm
\noindent Acknowledgments: I would like to thank Dr. Rob Brekelmans for many fruitful discussions on likelihood ratio exponential families and related topics.

\bibliographystyle{plain}
\bibliography{CILREFBib}

\appendix

\section{Exponential family of univariate Gaussian distributions}\label{sec:unigaussian}
Consider the family of univariate normal distributions:
$$
\calN=\left\{p_{\mu,\sigma^2}(x)=\frac{1}{\sqrt{2\pi\sigma^2}} \exp\left(-\frac{1}{2}\frac{(x-\mu)^2}{\sigma^2}\right), \mu\in\bbR,\sigma^2>0\right\}.
$$
Let $\lambda=(\lambda_1=\mu,\lambda_2=\sigma^2)$ denote the mean-variance parameterization, and consider the sufficient statistic vector $t(x)=(x,x^2)$. Then the densities of $\calN$ can be written in the canonical form of exponential families:
$$
p_{\lambda}(x)=\exp\left(\inner{\theta(\lambda)}{t(x)}-F(\theta)\right),
$$
where $\theta(\lambda)=\left(\frac{\lambda_1}{\lambda_2},-\frac{1}{2\lambda_2}\right)$ and the log-normalizer 
is
$$
F(\theta)=-\frac{\theta_1^2}{4\theta_2}+\frac{1}{2}\log \frac{\pi}{-\theta_2}.
$$
The dual moment parameterization is $\eta(\lambda)=E_{p_\lambda}[t(x)]=\left(\lambda_1,\lambda_1^2+\lambda_2\right)$, and the convex conjugate is:

$$
F^*(\eta)=\sup_{\theta\in\Theta}\{\inner{\theta}{\eta}-F(\eta)\}
=-\frac{1}{2}(\log(2\pi e (\eta_2-\eta_1^2)).
$$
We  check that the convex conjugate coincides with the negentropy~\cite{nielsen2010entropies}:
$$
h[p_\lambda]=-F^*(\eta(\lambda)).
$$

The conversion formul\ae{} between the dual natural/moment parameters and the ordinary parameters are given by:
\begin{eqnarray}
\theta(\lambda) &=& \left(\frac{\lambda_1}{\lambda_2},-\frac{1}{2\lambda_2}\right),\\
\lambda(\theta) &=& \left(-\frac{\theta_1}{2\theta_2},-\frac{1}{2\theta_2}\right),\\
\eta(\lambda) &=& \left(\lambda_1,\lambda_1^2+\lambda_2\right),\\
\lambda(\eta) &=& \left(\eta_1,\eta_2-\eta_1^2\right),\\
\eta(\theta) &=& (E[x],E[x^2]) = \nabla F(\theta) = \left(-\frac{\theta_1}{2\theta_2},-\frac{1}{2\theta_2}+\frac{\theta_1^2}{4\theta_2^2}\right) ,\\
\theta(\eta) &=&  \nabla F^*(\eta) =\left(-\frac{\eta_1}{\eta_1^2-\eta_2},\frac{1}{2(\eta_1^2-\eta_2)}\right) 
\end{eqnarray}

We check that
\begin{eqnarray*}
D_\KL[p_{\lambda}:p_{\lambda'}] &=&  \frac{1}{2}\left(\frac{(\lambda_1-\lambda_1')^2}{{\lambda'_2}^2}
+\frac{\lambda_2^2}{{\lambda_2'}^2}-\log \frac{\lambda_2^2}{{\lambda_2'}^2}-1 \right),\\
&=&B_F(\theta(\lambda'):\theta(\lambda))=B_{F^*}(\eta(\lambda):\eta(\lambda')),\\
&=& Y_{F,F^*}(\theta(\lambda'):\eta(\lambda))=Y_{F^*,F}(\eta(\lambda):\theta(\lambda')),
\end{eqnarray*}
where $B_F$ and $B_{F^*}$ are the dual Bregman divergences and $Y_{F,F^*}$ and $Y_{F^*,F}$ are the dual Fenchel-Young divergences.

\section{Code snippets in {\sc Maxima}}\label{sec:maxima}
Code for plotting Figure~\ref{fig:FBhatEx1}.
{\small
\begin{lstlisting}[language=Maxima,breaklines=true]
varalpha(v1,v2,alpha):=(v1*v2)/((1-alpha)*v1+alpha*v2)$
mualpha(mu1,v1,mu2,v2,alpha):=(alpha*mu1*v2+(1-alpha)*mu2*v1)/((1-alpha)*v1+alpha*v2)$
assume(v1>0)$assume(v2>0)$
theta1(mu,v):=mu/v$
theta2(mu,v):=-1/(2*v)$;
F(theta1,theta2):=((-theta1**2)/(4*theta2))+(1/2)*log(-%pi/theta2)$
JF(alpha,theta1,theta2,theta1p,theta2p):=alpha*F(theta1,theta2)+(1-alpha)*F(theta1p,theta2p)-F(alpha*theta1+(1-alpha)*theta1p,alpha*theta2+(1-alpha)*theta2p);

m1:0;v1:1;m2:1;v2:2;

plot2d([JF(alpha,theta1(m1,v1),theta2(m1,v1),theta1(m2,v2),theta2(m2,v2)),
-JF(alpha,theta1(m1,v1),theta2(m1,v1),theta1(m2,v2),theta2(m2,v2)),
 [discrete,[[0.4215580558605244,-0.15],[0.4215580558605244,0.15]]],
[discrete, [0.4215580558605244], [0.1155433222682347]],
[discrete, [0.4215580558605244], [-0.1155433222682347]]
],
[alpha,0,1], [xlabel,"alpha"], [ylabel,"F_{pq}(alpha)=-D_{B,alpha}[p:q]"],
[style,  [lines,1,1],[lines,1,2],
 [lines,2,0], [points, 3,3],[points, 3,3]],[legend, "skew Bhattacharyya D_{B,alpha}[p:q]","LREF log-normalizer F_{pq}(alpha)","","","" ],
[color, blue, red,   black, black,black],[point_type,asterisk]);
\end{lstlisting}
}

Code for calculating the Chernoff information between two univariate Gaussian distributions (Proposition~\ref{prop:unigauss}):

{\small
\begin{lstlisting}[language=Maxima,breaklines=true] 
varalpha(v1,v2,alpha):=(v1*v2)/((1-alpha)*v1+alpha*v2)$
mualpha(mu1,v1,mu2,v2,alpha):=(alpha*mu1*v2+(1-alpha)*mu2*v1)/((1-alpha)*v1+alpha*v2)$

/* Kullback-Leibler divergence */
KLD(mu1,v1,mu2,v2):=(1/2)*((((mu2-mu1)**2)/v2)+(v1/v2)-log(v1/v2)-1)$

assume(alpha>0)$assume(alpha<1)$
assume(v1>0)$assume(v2>0)$
theta1(mu,v):=mu/v$
theta2(mu,v):=-1/(2*v)$;
F(theta1,theta2):=-theta1**2/(4*theta2)+0.5*log(-1/theta2)$

eq: (theta1(mu1,v1)-theta1(mu2,v2))*mualpha(mu1,v1,mu2,v2,alpha)+(theta2(mu1,v1)-theta2(mu2,v2))*(mualpha(mu1,v1,mu2,v2,alpha)**2+varalpha(v1,v2,alpha))-F(theta1(mu1,v1),theta2(mu1,v1))+F(theta1(mu2,v2),theta2(mu2,v2));
solalpha: solve(eq,alpha)$
alphastar:rhs(solalpha[1]);

ChernoffInformation: KLD(mualpha(mu1,v1,mu2,v2,alphastar),varalpha(v1,v2,alphastar),mu1,v1)$
print("Chernoff information=")$ratsimp(ChernoffInformation);
\end{lstlisting}
}

Example of a plot of the $\alpha$-Bhattacharryya distance for $\alpha\in [0,1]$ when $p$ and $q$ are two normal distributions.
{\small
\begin{lstlisting}[language=Maxima,breaklines=true] 
varalpha(v1,v2,alpha):=(v1*v2)/((1-alpha)*v1+alpha*v2)$
mualpha(mu1,v1,mu2,v2,alpha):=(alpha*mu1*v2+(1-alpha)*mu2*v1)/((1-alpha)*v1+alpha*v2)$
assume(v1>0)$assume(v2>0)$
theta1(mu,v):=mu/v$
theta2(mu,v):=-1/(2*v)$;
F(theta1,theta2):=((-theta1**2)/(4*theta2))+(1/2)*log(-%pi/theta2)$
JF(alpha,theta1,theta2,theta1p,theta2p):=alpha*F(theta1,theta2)+(1-alpha)*F(theta1p,theta2p)-F(alpha*theta1+(1-alpha)*theta1p,alpha*theta2+(1-alpha)*theta2p);
m1:0;v1:1;m2:1;v2:2;
plot2d(JF(alpha,theta1(m1,v1),theta2(m1,v1),theta1(m2,v2),theta2(m2,v2)),[alpha,0,1]);
\end{lstlisting}
}

Example which calculates exactly the Chernoff exponent between two centered 4D Gaussians by solving the polynomial roots of the Chernoff optimal condition:
{\small
\begin{lstlisting}[language=Maxima,breaklines=true] 
assume(l1>0);assume(l2>0);assume(l3>0);assume(l4>0);
assume(alpha>0);assume(alpha<1);
l1:1;l2:2;l3:3;l4:4;
eq: (1-l1)/(alpha+(1-alpha)*l1)+ (1-l2)/(alpha+(1-alpha)*l2)+ (1-l3)/(alpha+(1-alpha)*l3)+ (1-l4)/(alpha+(1-alpha)*l4) + log(l1)+log(l2)+log(l3)+log(l4);
solve(eq,alpha);
sol:float(%);
realpart(sol);imagpart(sol);
/* alpha=0.5969427599369763 */
\end{lstlisting}
}

Example of choosing two different exponential families: The half-normal distributions and the exponential distributions:
{\small
\begin{lstlisting}[language=Maxima,breaklines=true] 
assume(sigma>0);
halfnormal(x,sigma):=(sqrt(2)/(sqrt(%pi*sigma**2)))*exp(-x**2/(2*sigma**2));
assume(lambda>0);
exponential(x,lambda):=lambda*exp(-lambda*x);
/* KLD diverges */
integrate(exponential(x,lambda)*log(exponential(x,lambda)/halfnormal(x,sigma)),x,0,inf);
/* KLD converges */
integrate(halfnormal(x,sigma)*log(halfnormal(x,sigma)/exponential(x,lambda)),x,0,inf);
/* Bhattacharyya coefficient */
assume(alpha>0);
assume(alpha<1);
integrate( (halfnormal(x,sigma)**alpha) * (exponential(x,lambda)**(1-alpha)),x,0,inf);  
\end{lstlisting}
}

\end{document}